\documentclass[jcs,crcready]{iosart1c}

\usepackage[T1]{fontenc}
\usepackage{times}%
\usepackage[numbers]{natbib}
\usepackage{amsmath,amssymb,amsthm,latexsym}
\usepackage{paralist}
\usepackage{booktabs,tabulary,multirow}
\usepackage{xspace}
\usepackage{url}
\usepackage{color}
\usepackage[ruled,lined,linesnumbered,boxed,vlined,algo2e]{algorithm2e}

\usepackage{pgfplots}
\pgfplotsset{compat=newest} 
\usepackage{caption}
\usepackage{subcaption}
\usepackage{tikz}
\usetikzlibrary{positioning,arrows,automata,shapes,decorations.pathreplacing}

\usepackage{booktabs}
\usepackage{subfig}
\usepackage{mathtools}
\usepackage{mathrsfs}

\newtheorem{thm}{Theorem}
\newtheorem{lemma}{Lemma}
\newtheorem{prop}{Proposition}

\newtheorem{cor}{Corollary}

\newtheorem{df}{Definition}

\newcommand{\set}[1]{\left\{#1\right\}}
\newcommand{\card}[1]{\left|#1\right|}
\newcommand{\etal}{\textit{et al.}\xspace}

\newcommand{\atmost}{\mathsf{atmost}}
\newcommand{\atleast}{\mathsf{atleast}}

\newcommand{\BOWSP}{\textsc{BO-WSP}\xspace}
\newcommand{\BOWSPLC}{\textsc{BO-WSP-LC}\xspace}
\newcommand{\BOWSPPO}{\textsc{BO-WSP-PO}\xspace}
\newcommand{\BOWSPPF}{\textsc{BO-WSP-PF}\xspace}

\SetKwInOut{Input}{input}
\SetKwInOut{Output}{output}
\newcommand{\PF}{\mathcal{P}}

\pubyear{0000}
\volume{0}
\firstpage{1}
\lastpage{1}

\begin{document}

\begin{frontmatter}    

\title{The bi-objective workflow satisfiability problem and workflow resiliency\thanks{A preliminary version of this work appeared in~\cite{CrGuKa15}. Details about differences between the two versions are given at the end of the Introduction section.}
}
\runningtitle{\BOWSP and applications}

\maketitle

\author[A]{\fnms{Jason} \snm{Crampton}\thanks{Corresponding author: Information Security Group, Royal Holloway, University of London, Egham, TW20 9QY, Egham; +44 1784 443117; \url{jason.crampton@rhul.ac.uk}}},
\author[A]{\fnms{Gregory} \snm{Gutin}},
\author[B]{\fnms{Daniel} \snm{Karapetyan}}, and
\author[A]{\fnms{R\'emi} \snm{Watrigant}}
\runningauthor{Crampton, Gutin, Karapetyan, Watrigant}
\address[A]{Royal Holloway, University of London}
\address[B]{University of Essex and University of Nottingham}

\begin{abstract}
 A computerized workflow management system may enforce a security policy, specified in terms of authorized actions and constraints, thereby restricting which users can perform particular steps in a workflow.
 The existence of a security policy may mean that a workflow is unsatisfiable, in the sense that it is impossible to find a valid plan (an assignment of steps to authorized users such that all constraints are satisfied).
 Work in the literature focuses on the workflow satisfiability problem, a \emph{decision} problem that outputs a valid plan if the instance is satisfiable (and a negative result otherwise).
 
 In this paper, we introduce the \textsc{Bi-Objective Workflow Satisfiability Problem} (\BOWSP), which enables us to solve \emph{optimization} problems related to workflows and security policies.
 In particular, we are able to compute a ``least bad'' plan when some components of the security policy may be violated.
In general, \BOWSP is intractable from both the classical and parameterized complexity point of view (where the parameter is the number of steps). 
We prove that computing a Pareto front for \BOWSP is fixed-parameter tractable (FPT) if we restrict our attention to user-independent constraints.
This result has important practical consequences, since most constraints of practical interest in the literature are user-independent.

Our proof is constructive and defines an algorithm, the implementation of which we describe and evaluate. 
We also present a second algorithm to compute a Pareto front which solves multiples instances of a related problem using mixed integer programming (MIP).
We compare the performance of both our algorithms on synthetic instances, and show that the FPT algorithm outperforms the MIP-based one by
several orders of magnitude on most instances.

Finally, we study the important question of workflow resiliency and prove new results establishing that known decision problems are fixed-parameter tractable when restricted to user-independent constraints.
We then propose a new way of modeling the availability of users and demonstrate that many questions related to resiliency in the context of this new model may be reduced to instances of \BOWSP.
\end{abstract}


\end{frontmatter}

\section{Introduction}\label{sec:intro}

It is increasingly common for organizations to computerize their business and management processes.
The co-ordination of the tasks or steps that comprise a computerized business process is managed by a workflow management system (or business process management system).
A workflow is defined by the steps in a business process and the order in which those steps should be performed.
A workflow is executed multiple times, each execution being called a \emph{workflow instance}.
Typically, the execution of each step in a workflow instance will be triggered by a human user, or a software agent acting under the control of a human user.
As in all multi-user systems, some form of access control, typically specified in the form of policies and constraints, should be enforced on the execution of workflow steps, thereby restricting the execution of each step to some authorized subset of the user population.

Policies typically specify the workflow steps for which users are authorized (what Basin \etal\ call \emph{history-independent} authorizations~\cite{BaBuKa12}).
Constraints restrict which groups of users can perform sets of steps.
It may be that a user, while authorized by the policy to perform a particular step $s$, is prevented (by one or more constraints) from executing $s$ in a specific workflow instance because particular users have performed other steps in the workflow (hence the alternative name of \emph{history-dependent} authorizations~\cite{BaBuKa12}).
The concept of a Chinese wall, for example, limits the set of steps that any one user can perform~\cite{BrNa89}, as does separation-of-duty, which is a central part of the role-based access control model~\cite{ansi-rbac04}.
We note that policies are, in some sense, discretionary, as they are defined by the workflow administrator in the context of a given set of users.
However, constraints may be mandatory (and independent of the user population), in that they may encode statutory requirements governing privacy, separation-of-concerns, or high-level organizational requirements.

A simple, illustrative example for purchase order processing~\cite{CrGu13} is shown in Figure~\ref{fig:example-workflow}.
In the first step of the workflow, the purchase order is created and approved (and then dispatched to the supplier).
The supplier will submit an invoice for the goods ordered, which is processed by the create payment step.
When the supplier delivers the goods, a goods received note (GRN) must be signed and countersigned.
Only then may the payment be approved and sent to the supplier. 

In addition to defining the order in which steps must be performed, the workflow specification includes constraints to prevent fraudulent use of the purchase order processing system.
In our purchase order example, these constraints restrict the users that can perform pairs of steps in the workflow (illustrated in Figure~\ref{subfig:constraints}): the same user may not sign and countersign the GRN, for example.
In addition, we may specify an authorization policy, in this case using roles (Figures~\ref{subfig:role-step-assignment} and~\ref{subfig:user-role-assignment}).

\begin{figure}[h]\centering
\begin{subfigure}[b]{.325\textwidth}\centering
\includegraphics{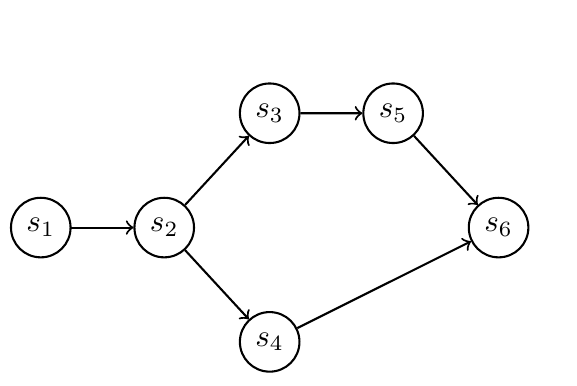}
\caption{Ordering on steps}\label{subfig:hasse}
\end{subfigure}
\hfill
\begin{subfigure}[b]{.15\textwidth}\centering
\includegraphics{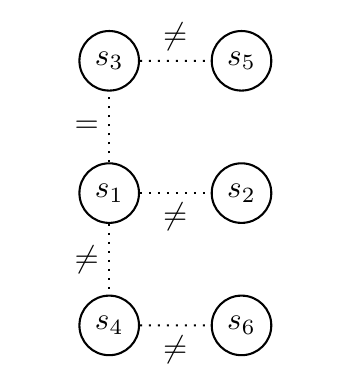}

\caption{Constraints}\label{subfig:constraints}
\end{subfigure}
\hfill
\begin{subfigure}[b]{.4\textwidth}\centering\small\setlength{\extrarowheight}{1pt}
  \begin{tabular}{|@{}l@{}@{}l@{}|}
    \hline
    $s_1$ & create purchase order \\
    $s_2$ & approve purchase order \\
    $s_3$ & sign GRN \\
    $s_4$ & create payment \\
    $s_5$ & countersign GRN \\
    $s_6$ & approve payment \\
    \hline
    $\ne$ & different users must perform steps \\
    $=$ & same user must perform steps \\
    \hline
  \end{tabular}
 \caption{Legend}
\end{subfigure}

\begin{subfigure}[b]{.35\textwidth}\centering\small
 \[
  \begin{array}{|lr|rrrrrr|}
  \hline
   & & s_1 & s_2 & s_3 & s_4 & s_5 & s_6 \\
  \hline
   {\sf financeClerk} & r_1 & 1 & & 1 & 1 & & \\
   {\sf marketingClerk} & r_2 & 1 & & 1 & & & \\
   {\sf marketingAdmin} & r_3 & & 1 & 1 & & 1 & \\
   {\sf financeAdmin} & r_4 & & & 1 & 1 & 1 & \\
   {\sf accountant} & r_5 & & & & & 1 & 1 \\
  \hline
 \end{array}  
 \]
 \caption{Role-step assignment}\label{subfig:role-step-assignment}
\end{subfigure}
\hfill
\begin{subfigure}[b]{.225\textwidth}\centering\small
 \[
  \begin{array}{|r|rrrrr|}
   \hline
    & r_1 & r_2 & r_3 & r_4 & r_5 \\
   \hline
    u_1, u_2 & 1 & & & & \\
    u_3, u_4, u_5 & & 1 & & &  \\
    u_6 & & & 1 & & \\
    u_7 & & & & 1 & \\
    u_8 & & & & & 1\\
   \hline
  \end{array}
 \]
 \caption{User-role assignment}\label{subfig:user-role-assignment}
\end{subfigure}
\hfill
\begin{subfigure}[b]{.25\textwidth}\small
 \[
  \begin{array}{|r|rrrrrr|}
  \hline
   & s_1 & s_2 & s_3 & s_4 & s_5 & s_6 \\
  \hline
   u_1, u_2 & 1 & & 1 & 1 & & \\
   u_3, u_4, u_5 & 1 & & 1 & & & \\
   u_6 & & 1 & 1 & & 1 & \\
   u_7 & & & 1 & 1 & 1 & \\
   u_8 & & & & & 1 & 1 \\
  \hline
  \end{array}
 \]
 \caption{Authorization relation}\label{subfig:authorization-relation}
\end{subfigure}
\caption{A simple constrained workflow for purchase order processing}\label{fig:example-workflow}
\end{figure}

It is well known that a workflow specification may be ``unsatisfiable'' in the sense that the combination of policies and constraints means that there is no way of allocating authorized users to workflow steps without violating at least one constraint. Informally, an instance of the Workflow Satisfiability Problem (WSP) asks whether a particular workflow specification is satisfiable or not. 
This problem is NP-hard~\cite{WaLi10}, although relatively efficient algorithms have been developed~\cite{CoCrGaGuJo14,CrGuYe13,KaGaGu,WaLi10} (assuming that the number of workflow steps is much smaller than the number of users that may perform steps in the workflow and that all constraints are user-independent).%
\footnote{The satisfaction of a user-independent constraint~--~defined formally in Section~\ref{sec:related-work}~--~does not depend on the identity of the users.}

Of course, the objectives of the business process associated with a workflow specification can never be achieved if the WSP instance is unsatisfiable.
Hence, it is interesting to consider an extended version of WSP that seeks a ``best'' allocation of users to steps in the event that the instance is unsatisfiable.
We might, for example, decide that the ``cost'' of breaking a binding-of-duty constraint is lower than that associated with breaking a separation-of-duty constraint, and wish to find a plan, based on this assumption, in the event that the workflow is not satisfiable.
Alternatively, if we have to assign a step to an unauthorized user, we might well wish to associate different costs with those unauthorized users: in our purchase order example, the cost of assigning the create-payment step to a user in the \textsf{accountant} role will presumably be less than that of assigning the step to a user in the \textsf{marketingClerk} role.

Accordingly, Crampton, Gutin and Karapetyan introduced the Valued WSP~\cite{CrGuKa15} in which each plan is associated with two costs, which are determined, respectively, by the authorization policy violations and constraint violations induced by the plan. 
To make the Valued WSP a single-objective optimization problem, the two costs were aggregated into one total cost by summing the two costs. 
For a satisfiable instance of WSP there will exist at least one plan with zero cost.
However, if an instance is unsatisfiable, Valued WSP will return a plan of minimum total cost.

While this is a step in the right direction, the output of Valued WSP is rather crude, in that it simply generates a plan that minimizes the total cost.
In many situations, it may be appropriate to consider the cost of authorization policy violations and constraint violations separately.
As we noted above, authorization policy is typically a matter of local, "discretionary" configuration, whereas constraints may be mandated by statutory requirements.
Thus, we may want to find a plan that, for example, ensures the cost of constraint violations is zero and and the cost due to policy violations is minimized.

In this paper, we study the \emph{Bi-Objective} Workflow Satisfiability Problem (\BOWSP) with (weighted) user-independent constraints.
Informally, we associate each constraint and authorization violation with a cost, which may, for example, be regarded as an estimate of the risk associated with allowing those violations.
We then consider \BOWSP as a bi-objective optimization problem and identify three approaches to \BOWSP optimization, including the computation of a so-called Pareto front that generalizes the other two.
A Pareto front is a maximal set of incomparable \BOWSP solutions, allowing practitioners to choose the most suitable of those solutions. 
In a sense, our work is related to recent work on risk-aware access control~\cite{ChCr11,ChRoKeKaWaRe07,DiBeEyBaMo04,NiBeLo10}, which seeks to compute the risk of allowing a user to perform an action, rather than simply computing an allow/deny decision, and ensure that cumulative risk remains within certain limits.
However, unlike related work, we focus on computing user-step assignments of minimal costs, rather than access control decisions.

Workflow resiliency is a property related to workflow satisfiability~\cite{WaLi10}: informally, a workflow specification is resilient if it remains satisfiable even if some authorized users are not available to perform steps.
It is known that decision problems related to resiliency are PSPACE-complete~\cite[Theorems 14--15]{WaLi10}.
Recent work has investigated quantitative resiliency, which seeks to measure the extent to which a workflow is resilient, rather than simply deciding whether it is resilient or not~\cite{MaMoMo14}.
In this paper, we study the decision problems of Wang and Li~\cite{WaLi10} from the perspective of fixed-parameter tractability.
We also demonstrate that some questions regarding quantitative resiliency can be reduced to \BOWSP.

We now summarize the main contributions of our paper:
\begin{itemize}
  \item to define \BOWSP as a bi-objective optimization version of WSP;
  \item to prove that computing a Pareto front admits a fixed-parameter tractable (FPT) algorithm\footnote{fixed-parameter tractability is discussed in the next section.} for weighted user-independent constraints;
  \item to prove that the running time is the best we can hope for in theory (see Theorem \ref{thm:tightness} and remarks afterwards);
  \item to show how to implement this result efficiently by presenting an algorithm called \emph{Pattern Branch and Bound (PBB)};
  \item to provide a second algorithm for computing a Pareto front for \BOWSP, based on mixed-integer programming (MIP);
  \item to prove that \BOWSP can be used to solve (a generalized version of) the Cardinality-Constrained Minimum User Problem~\cite{Roy2015};
  \item to provide a parameter-driven instance generator for \BOWSP, which may be of independent interest;
  \item to use this generator to compare the performance of our two methods for solving \BOWSP, and conclude that PBB demonstrates significantly better performance on most of the test instances;
  \item to prove that decision problems for workflow resiliency that are known to be PSPACE-complete~\cite{WaLi10} are FPT;
  \item to introduce a new model for reasoning about workflow resiliency, and 
  \item to show how \BOWSP can be used to solve problems related to workflow resiliency. 
\end{itemize}

In the next section, we review relevant concepts from the literature concerning the workflow satisfiability problem, user-independent constraints and resiliency in workflows.
In Section~\ref{sec:wwsp}, we define the bi-objective problem \BOWSP, and describe three ways in which the problem could be addressed. 
Then we give a fixed-parameter algorithm that can be used to solve \BOWSP in the case of user-independent constraints. We also describe a generalization of a problem introduced by Roy {\em et al.}~\cite{Roy2015} and how our fixed-parameter algorithm can be applied to the generalized problem.
In Section~\ref{sec:experiments}, we describe the implementation of two algorithms for computing a Pareto front of \BOWSP, one relying on the FPT approach described in the previous section, the other using mixed integer programming. 
We then describe an experimental framework for generating random instances of \BOWSP, and compare the performance of the two implementations.
In Section~\ref{sec:resiliency}, we study the resiliency problem in workflow systems. We review the model of resiliency of Wang and Li~\cite{WaLi10} and give a fixed-parameter algorithm deciding whether a workflow is resilient up to $t$ unavailable users, when only UI constraints are considered. We then define a new model allowing an organization to encode the availabilities of users in terms of probabilities, and then reduce this problem to \BOWSP in order to find a plan with the highest  likelihood of success, using the algorithm presented in Section~\ref{sec:wwsp}.
We conclude the paper with a discussion of related work, a summary of contributions and some suggestions for future work.\\

A preliminary version of this work was presented at the $20^{th}$ ACM Symposium on Access Control Models and Technologies (SACMAT 2015) \cite{CrGuKa15}. The main additional contributions of the present paper are \begin{inparaenum}[(i)]%
 \item to formally state the \BOWSP as a bi-objective optimization problem (only a single-objective problem, a special case of \BOWSP, was introduced in \cite{CrGuKa15});
 \item to give a fixed-parameter algorithm and an MIP-based algorithm for computing a Pareto front for \BOWSP, thus generalizing the algorithms of \cite{CrGuKa15};
\item to compare experimentally the performances of both algorithms using a richer generator of synthetic test instances;
 \item to study resiliency in workflow systems, first as a decision problem, by giving the first positive results (using fixed-parameter algorithms) with respect to the model of \cite{WaLi10}, and secondly by introducing a new model relying on probabilities of user availabilities. We then illustrate one application of \BOWSP by showing that the proposed algorithm can be used to find a most resilient plan with respect to this new model. 
\end{inparaenum}

%

\section{Definitions and related work}\label{sec:related-work}

In this section we recall the definition of the workflow satisfiability problem and user-independent constraints.
We also review related work.

\subsection{The workflow satisfiability problem}

\begin{df}\label{def:workflow}
A \emph{workflow specification} is defined by a partially ordered set $(S,<)$, where $S$ is a set of steps.
Given a workflow specification $(S,<)$ and a set of users $U$, an \emph{authorization policy} for a workflow specification is a relation $A \subseteq S \times U$.
A \emph{workflow authorization schema} is a tuple $((S,<),U,A)$, where $(S,<)$ is a workflow specification, $U$ the set of users, and $A$ an authorization policy.
\end{df}

A workflow specification describes a set of steps and restricts the order in which they must be performed when the workflow is executed, each such execution being called a \emph{workflow instance}.
In particular, if $s < s'$ then $s$ must be performed before $s'$ in every instance of the workflow.
User $u$ is authorized\footnote{In practice, the set of authorized step-user pairs, $A$, will not be defined explicitly.
Instead, $A$ will be inferred from other access control data structures.
In our purchase order workflow, for example, it is straightforward to infer $A$ (Figure~\ref{subfig:authorization-relation}) from the role-step and user-role assignment relations (Figures~\ref{subfig:role-step-assignment} and~\ref{subfig:user-role-assignment}).
We prefer to use $A$ in order to simplify the exposition.} to perform step $s$ only if $(s,u) \in A$.
We assume that for every step $s \in S$ there exists some user $u \in U$ such that $(s,u) \in A$ (otherwise the workflow is trivially unsatisfiable).

\begin{df}
A \emph{workflow constraint} has the form $(T,\Theta)$, where $T \subseteq S$ and $\Theta$ is a family of functions with domain $T$ and range $U$.
$T$ is the \emph{scope} of the constraint $(T,\Theta)$.
A \emph{constrained workflow authorization schema} is a tuple $((S,<),U,A,C)$, where $((S,<),U,A)$ is a workflow authorization schema and $C$ is a set of workflow constraints.
\end{df}
In the example of Figure \ref{fig:example-workflow}, the constraint requiring that $s_1$ and $s_3$ are assigned the same user can be written as $(T', \Theta')$, where $T'=\{s_1,s_3\}$ and  $\Theta'= 
\{T' \rightarrow \{u_i\}:\ i=1,2,\dots ,8\}.$

Given $\pi :S' \rightarrow U$, where $S' \subseteq S$, and $T \subseteq S'$, we define $\pi|_T : T \rightarrow U$, where $\pi|_T(s) = \pi(s)$ for all $s \in T$.

\begin{df}
Let $((S,<),U,A,C)$ be a constrained workflow authorization schema.
A \emph{(partial) plan} is a function $\pi: S' \rightarrow U$, where $S'\subseteq S$. A plan $\pi$ is \emph{complete} if $S'=S.$
Given $S' \subseteq S$, we say a plan $\pi : S' \rightarrow U$: 
\begin{itemize}
 \item is \emph{authorized} if for all $t \in S'$, $(t,\pi(t)) \in A$;
 \item \emph{satisfies} a workflow constraint $(T,\Theta)$ if $T \not\subseteq S'$ or $\pi|_T \in \Theta$; and
 \item is \emph{valid} if it is authorized and satisfies every constraint in $C$.
\end{itemize}
\end{df}

Informally, a workflow constraint $(T,\Theta)$ limits the users that are allowed to perform a set of steps $T$ in any given instance of the workflow.
In particular, $\Theta$ identifies authorized (partial) assignments of users to workflow steps in $T$.
The objective of the workflow satisfiability problem is to find a valid complete plan (if one exists).
More formally, we have the following definition~\cite{WaLi10}:

\begin{center}
\fbox{%
      \begin{tabulary}{.95\columnwidth}{@{}r<{~}@{}L@{}}
        \multicolumn{2}{@{}l}{\sc Workflow Satisfiability Problem (WSP)}\\
        \emph{Input:} & A constrained workflow authorization schema $((S,<),U,A,C)$\\
        \emph{Output:} & A valid plan $\pi : S \rightarrow U$ (WSP is satisfiable) or an answer that WSP is not satisfiable
       \end{tabulary}%
      }
\end{center}

Consider a pair $\tau,\tau'$ of  complete plans in our purchase order workflow defined as follows: $\tau(s_1)=\tau(s_3)=u_1$, $\tau(s_2)=\tau(s_5)=u_6,$ $\tau(s_4)=u_2$, $\tau(s_6)=u_8$, $\tau'(s_1)=\tau'(s_3)=u_2$, $\tau'(s_2)=\tau'(s_5)=u_6,$ $\tau'(s_4)=u_1$, and $\tau'(s_6)=u_7$. Since $\tau$ is a valid plan, the purchase order workflow is satisfiable. Plan $\tau'$ is not valid as $u_7$ is not authorized for $s_6$. If we add the constraint that $s_5$ and $s_6$ must be performed by different users and
$u_6$ and $u_7$ become unavailable, then the workflow becomes unsatisfiable as only $u_8$ is authorized to perform either of $s_5$ and $s_6$.

WSP may be used in a number of ways. 
First, it may be used as a form of static analysis, prior to deployment, to ensure that the workflow specification is useful, in the sense that there exists at least one ``execution path'' through the workflow. 
Second, it may be used to generate plans for workflow instances, where a plan assigns steps to users for each instance of the workflow specification when it is instantiated. 
Finally, it may be used in a more dynamic way, when steps in a workflow instance are not assigned to users in advance. Instead, a user would select a ``ready'' step in a workflow instance to execute. A ready step is any step $t$ such that all steps $s$ such that $s < t$ have been performed. 
We may model a partially complete workflow instance as a workflow specification in which the set of authorized users for any step that has been performed is a singleton, containing the user that performed the step.
Then a request by user $u$ to perform $t$ is evaluated by the workflow management system: the request is permitted if:
\begin{inparaenum}[(i)]%
 \item $u$ is authorized to perform $t$; and
 \item the updated workflow specification in which the only authorized user for step $t$ is $u$ is satisfiable.
\end{inparaenum}

It is worth noting that constraints are defined in terms of sets and functions, and a plan is also defined to be a function. That is, constraint satisfaction is (assumed to be) independent of the ordering on steps. 
Moreover, the authorization policy is fixed and independent of the ordering on steps. In other words, the ordering on steps, for the ``standard'' applications of WSP described above, is only relevant to the sequence in which the steps of a workflow instance are executed, not in determining whether there exists a valid plan (or the construction of a valid plan). 
However, in Section~\ref{sec:resiliency}, we will see that we may model workflow resiliency problems using authorization policies that vary over time. In this context, the order in which steps are performed is relevant to satisfiability.

\subsection{Constraint specification and user-independent constraints}
\label{sec:ui-constraints}

In practice, we do not define constraints by enumerating all possible elements of $\Theta$.
Instead, we define families of constraints that have ``compact'' descriptions.
Thus, for example, we might define the family of simple separation-of-duty constraints, each of which could be represented by a set $\set{t_1,t_2} \subseteq S$, the constraint being satisfied by a plan $\pi$ if and only if $\pi(t_1) \neq \pi(t_2)$.
More generally, for any binary relation $\propto$ defined on $U$, constraint $(s_1,s_2,\propto)$ is satisfied by any plan $\pi$ such that $\pi(s_1) \propto \pi(s_2)$.
(In particular, we may represent a binding-of-duty constraint in the form $(s_1,s_2,=)$ and a separation-of-duty constraint in the form $(s_1,s_2,\ne)$.)

A constraint $(T,\Theta)$ is said to be \emph{user-independent} (UI, for short) if, for every $\theta \in \Theta$ and every permutation $\phi : U \rightarrow U$, we have $\phi \circ \theta \in \Theta$, where $(\phi \circ \theta)(s) = \phi(\theta(s))$.
Informally, a constraint is UI if it does not depend on the identity of the users.
It appears most constraints that are useful in practice are UI~\cite{CoCrGaGuJo14}.
In particular, cardinality constraints, separation-of-duty and binding-of-duty constraints are UI.
In our experiments (Section~\ref{sec:experiments}), we will consider two particular types of UI constraints:
\begin{itemize}
 \item an \emph{at-least-$r$} \emph{counting} constraint has the form $\atleast(T,r)$, where $r \leqslant \card{T}$, and is satisfied provided at least $r$ users are assigned to the steps in $T$;
 \item an \emph{at-most-$r$} \emph{counting} constraint has the form $\atmost(T,r)$, where $r \leqslant \card{T}$, and is satisfied provided at most $r$ users are assigned to the steps in $T$.
\end{itemize}
Note that a separation-of-duty constraint $(s_1,s_2,\ne)$ is the counting constraint $\atleast(\set{s_1,s_2},2)$, and a binding-of-duty constraint $(s_1,s_2,=)$ is the constraint $\atmost(\set{s_1,s_2},1)$.

It is important to stress that our approach works for any UI constraints.
We chose to use counting constraints because such constraints have been widely considered in the literature (often known as cardinality constraints~\cite{CoCrGaGuJo14}).
Moreover, counting constraints can be efficiently encoded using mixed integer programming, so we can use off-the-shelf solvers to solve WSP and thus compare their performance with that of our bespoke algorithms.

\subsection{Related work}

A na\"ive approach to solving WSP would consider every possible assignment of users to steps in the workflow.
There are $n^k$ such assignments if there are $n$ users and $k$ steps, so an algorithm of this form would have complexity $\Omega(mn^k)$, where $m$ is the number of constraints.\footnote{We say that $f = \Omega(g)$ for two functions $f, g : \mathbb{N} \rightarrow \mathbb{N}$ if there exist $N_0, d \in \mathbb{N}^*$ such that $f(N) \ge d g(N)$ for all $N \ge N_0$.}
Moreover, Wang and Li showed that WSP is NP-hard, by reducing {\sc Graph $k$-Colorability} to WSP with separation-of-duty constraints~\cite[Lemma 3]{WaLi10}.

\subsubsection{Fixed-parameter tractability results for WSP}
\label{sec:fpt-results}

The importance of finding an efficient algorithm for solving WSP led Wang and Li to look at the problem from the perspective of \emph{parameterized complexity} \cite{DowFel13,Ni06}.
Parameterized complexity is a two-dimensional extension of classical complexity: in parameterized complexity an instance of a decision problem is a pair consisting of an instance $x\in \Sigma^*$ of a classical decision problem and a parameter $\kappa\in \mathbb{N}$, where $\Sigma$ is an alphabet. An algorithm is said to be \emph{fixed-parameter tractable} (FPT) if it solves a problem in time $O(f(\kappa)p(|x|))$, where $|x|$ is the size of $x$, and $f$ and $p$ are respectively a computable function and a polynomial. Normally, for applications $\kappa$ is supposed to be relatively small and thus if $f$ is a moderately growing function, an FPT algorithm may provide an efficient solution for the problem under consideration. If a problem can be solved using an FPT algorithm then we say that it is an \emph{FPT problem} and that it belongs to the class FPT.

Let us introduce a class W[1] of parameterized problems that includes all FPT problems, but whose hardest problems (W[1]-complete problems) are widely believed not to be in FPT \cite{DowFel13,Ni06}. Consider two parameterized problems $P$ and $P'$ (each problem is a set of two-dimensional instances, \emph{i. e.} a subset of $\Sigma^* \times \mathbb{N}$, where $\Sigma$ is a fixed alphabet). We say that $P$ reduces to $P'$ by a {\em parameterized reduction} if for every $(x, \kappa) \in \Sigma^* \times \mathbb{N}$, one can compute $(x', \kappa') \in \Sigma^* \times \mathbb{N}$ in time $O(f(k)(|x|+\kappa)^c)$ for a computable function $f$ and a constant $c$ such that (i) $\kappa' \le g(\kappa)$ for a computable function $g$, and (ii) $(x, \kappa) \in P$ if and only if $(x', \kappa') \in P'$.
Class W[1] consists of all parameterized problems that can be reduced to {\sc Independent Set} parameterized by the size of the solution: given a graph $G$ and a parameter $\kappa$, decide whether $G$ has an independent set of size $\kappa$.


Wang and Li observed that parameterized complexity is an appropriate way to study the WSP, because the number $k$ of steps is usually small and often much smaller than the number $n$ of users;\footnote{The SMV loan origination workflow studied by Schaad \etal, for example, has 13 steps and identifies five roles~\cite{ScLoSo06}. It is generally assumed that the number of users is significantly greater than the number of roles.} thus, $k$ can be considered as the parameter. 
Wang and Li \cite{WaLi10}  proved that, in general, WSP is W[1]-hard and thus is highly unlikely to admit an FPT algorithm.
However, they also proved WSP is FPT if we consider only separation-of-duty and binding-of-duty constraints \cite{WaLi10}. Henceforth, we consider special families of constraints, but allow arbitrary authorizations.
Crampton \etal~\cite{CrGuYe13} obtained significantly faster FPT algorithms that were applicable to ``regular'' constraints, thereby including the cases shown to be FPT by Wang and Li.
Subsequent research has demonstrated the existence of FPT algorithms for {\sc WSP} in the presence of other constraint types~\cite{CrCrGuJoRa13,CrGu13}. Cohen \etal~\cite{CoCrGaGuJo14}
introduced the class of UI constraints and showed that WSP remains FPT if only UI constraints are included. Note that every regular constraint is UI and all the constraints defined in the ANSI RBAC standard \cite{ansi-rbac04} are UI. Results of Cohen \etal~\cite{CoCrGaGuJo14}  have led to algorithms which significantly outperform the widely used SAT-solver SAT4J on difficult instances of WSP with UI constraints \cite{CoCrGaGuJo14c,KaGaGu}.

\subsubsection{Finding plans for unsatisfiable instances}
 

There has been considerable interest in recent years in making the specification and enforcement of access control policies more flexible.
Naturally, it is essential to continue to enforce effective access control policies.
Equally, it is recognized that there may well be situations where a simple ``allow'' or ``deny'' decision for an access request may not be appropriate.
It may be, for example, that the risks of refusing an unauthorized request are greater than the risks of allowing it.
One obvious example occurs in healthcare systems, where the denial of an access request in an emergency situation could lead to loss of life.
Hence, there has been increasing interest in context-aware policies, such as ``break-the-glass''~\cite{BrPe09}, which allow different responses to requests in different situations.
Risk-aware access control is a related line of research that seeks to quantify the risk of allowing a request~\cite{DiBeEyBaMo04}, where a decision of ``0'' might represent an unequivocal ``deny'' and ``1'' an unequivocal ``allow'', with decisions of intermediate values representing different levels of risk.

Similar considerations arise very naturally when we consider workflows.
In particular, we may specify authorization policies and constraints that mean a workflow specification is unsatisfiable.
Clearly, this is undesirable from a business perspective, since the business objective associated with the workflow can not be achieved.
There are two possible ways of dealing with an unsatisfiable workflow specification: modify the authorization policy and/or constraints; or find a ``least bad'' complete plan.
Prior work by Basin, Burri and Karjoth considered the former approach~\cite{BaBuKa12}.
They restricted their attention to modification of the authorization policy, what they called \emph{administrable authorizations}.
They assigned costs to modifying different aspects of a policy and then computed a strategy to modify the policy of minimal cost.
We adopt a different approach and consider minimizing the cost of ``breaking'' the policies and/or constraints.

\section{The bi-objective workflow satisfiability problem}\label{sec:wwsp}

Informally, given a workflow specification and a plan $\pi$, we assign two costs to $\pi$, determined by the authorizations and constraints that $\pi$ violates. 
The {\sc Bi-objective Workflow Satisfiability Problem} (\BOWSP) seeks to find a plan that minimizes these costs.
WSP is a decision problem; \BOWSP is a multi-objective optimization problem.
Since it is not possible, in general, to minimize both costs in \BOWSP simultaneously, we consider three different approaches to solving \BOWSP.
We show that all three approaches can be tackled by FPT algorithms for user-independent constraints (Theorem~\ref{thm1}).\footnote{In Section \ref{sec:fpt-results}, we defined FPT only for the decision problems, but the notion can be naturally extended to optimization problems. In case of \BOWSP, an FPT algorithm runs in time $O(f(k) {\rm poly}(n)),$ where $k$ and $n$ are the numbers of steps and users, respectively.}

More formally, let $(S, <)$ be a workflow, $U$ a set of users, and $C$ a
set of constraints. Let $\Pi$ denote the set of all possible plans from $S$ to $U$, and let $\mathbb{Q}^+$ be the set of non-negative rationals.
For each $c \in C$, we say that a function $\omega_c : \Pi \rightarrow
\mathbb{Q}^+$ is a \emph{constraint weight function} if, for all $\pi \in \Pi$:
 \[
  \omega_c(\pi)
   \begin{cases}
    = 0 & \text{if $\pi$ satisfies $c$}, \\
    > 0 & \text{otherwise}.
   \end{cases}
 \]
The pair $(c,\omega_c)$ is a \emph{weighted constraint}.
The intuition is that $\omega_c(\pi)$ represents the extent to which $\pi$ violates $c$.

Consider the constraints $c_1 = (\set{s_1,s_2},\ne)$ and $c_2 = (\set{s_1,s_3},=)$ from our purchase order workflow.
Then we might define, for some $x \in \mathbb{Q}^+$,
\[ 
 \omega_{c_1}(\pi) =%
  \begin{cases}
   0 & \text{if $\pi(s_1) \ne \pi(s_2)$}, \\
   10x & \text{otherwise};
  \end{cases}
  \quad\text{and}\quad
 \omega_{c_2}(\pi) =%
  \begin{cases}
   0 & \text{if $\pi(s_1) = \pi(s_3)$}, \\
   x & \text{otherwise}.
  \end{cases}
\]
In this example, we have assigned a greater weight to breaking the separation-of-duty constraint than breaking the binding-of-duty constraint.
These weights reflect the danger of fraud that exists when the same user performs $s_1$ and $s_2$, compared to the relatively unimportant constraint that the same user creates a purchase order and signs for the goods received.
In a different workflow specification in which different security considerations are relevant, these weights might be reversed.

Now consider the counting constraint $\atmost(T,r)$.
Then $\omega_c(\pi)$ depends only on the number of users assigned to the steps in $T$ (and the penalty should increase as the number of users increases).
Let $\pi(T)$ denote the set of users assigned to steps in $T$.
Then $\omega_c(\pi) = 0$ if $\card{\pi(T)} \le r$, and $\omega_c(\pi) > 0$ otherwise. Further, if $\pi$ and $\pi'$ are two plans such that $\card{\pi(T)} < \card{\pi'(T)}$, then $\omega_c(\pi) \le \omega_c(\pi')$.

Having defined $\omega_c$ for each constraint in the workflow specification, we define
\[
 \omega_C(\pi) = \sum_{c \in C} \omega_c(\pi),
\]
which we call the \emph{constraint weight} of $\pi$.
Note that $\omega_C(\pi) = 0$ if and only if $\pi$ satisfies all constraints in $C$.
Note also that $\omega_C(\pi)$ need not be defined to be the linear sum: $\omega_C(\pi)$ may be defined to be an arbitrary function computable in time polynomial in the input size of the tuple $(\omega_c(\pi):\ c\in C)$ and Theorem \ref{thm1} below would still hold. However, we will not use this generalization in this paper, but simply remark that it is possible, if needed.

We say that a function $\omega: 2^S \times U \rightarrow \mathbb{Q}^+$ is a \emph{weighted set-authorization} function if $\omega(\emptyset, u) = 0$ for all $u \in U$. 
Let $\pi^{-1}(u) = \{ s \in S: \pi(s)=u \}$ denote the set of steps to which $u$ is assigned by plan $\pi$.
Then, given a weighted set-authorization function $\omega$, we define an \emph{authorization weight} function $\omega_A : \Pi \rightarrow \mathbb{Q}^+$ to be, for all $\pi \in \Pi$:
\[
 \omega_A(\pi) = \sum_{u \in U}  \omega(\pi^{-1}(u),u).
\]
In the remainder of the paper, we assume that both $\omega_A$ and $\omega_C$ are computable in polynomial time.

Given a workflow $(S,<)$, a set of users $U$, a set of weighted constraints $C_{\omega}$, and weighted set-authorization function $\omega$, we say that $((S,<),U,\omega,C_{\omega})$ is a \emph{weighted constrained workflow schema}.

\begin{center}
\fbox{%
      \begin{tabular}{p{0.08\linewidth}p{0.65\linewidth}}
        \multicolumn{2}{@{}l}{\sc Bi-Objective WSP (\BOWSP)}\\
        \emph{Input:} & A weighted constrained workflow schema $((S,<),U,\omega,C_{\omega})$ \\ 
        \emph{Output:} & A plan $\pi : S \rightarrow U$\\ 
        \emph{Goal:} & Minimize $\omega_C(\pi)$ and $\omega_A(\pi)$\\
       \end{tabular}%
      }
\end{center}

The alert reader will have noticed that the authorization relation is no longer included in the workflow specification.
This is because any authorization relation $A$ can be encoded by an appropriate choice of $\omega$.
The next result formalizes this claim and establishes that \BOWSP generalizes WSP.

\begin{prop}\label{pro}
 Given an instance of WSP $((S,<),U,A,C)$, let $\omega$ be a weighted set-authorization function such that for all $T \subseteq S$ and all $u \in U$:
  \[
   \omega(T,u)
    \begin{cases}
     = 0 & \text{if $(u,t) \in A$ for all $t \in T$}, \\
     > 0 & \text{otherwise};
    \end{cases}
  \]
  And, for each $c \in C$, let $\omega_c$ be a constraint weight function.
 Then the \BOWSP instance $((S,<),U,\omega,C_{\omega})$, where $C_{\omega} = \set{(c,\omega_c) : c \in C}$, has a plan $\pi$ such that $\omega_A(\pi) = \omega_C(\pi) = 0$ if and only if $((S,<),U,A,C)$ is satisfiable.
\end{prop}

In many cases, a weighted set-authorization function $\omega: 2^S \times U \rightarrow \mathbb{Q}^+$ satisfies the ``linearity'' condition $\omega(T,u)=\sum_{t\in T}\omega(\{t\},u)$ for every $T\subseteq S$ and $u\in U$.
However, $\omega$ does not necessarily satisfy this condition. For example:
\begin{enumerate}
 \item We can introduce a large penalty $\omega(T,u)$, effectively modeling simple ``prohibitions'' or ``negative authorizations'' (for instance, that we prefer not to involve $u$ in steps in $T$).
We use weights like this in our experimental work, described in Section~\ref{sec:experiments}.

In our purchase order example, we might set different (non-zero) weights for $\omega(\set{s_5},u_3)$, $\omega(\set{s_5},u_4)$ and $\omega(\set{s_5},u_5)$, depending on the relative seniority and/or experience of these three users (each of whom is assigned to the \textsf{marketingClerk} role).
 \item We can define a limit $\ell_u$ on the number of steps that can be executed by $u$, by setting a large penalty $\omega(T,u)$ for all $T$ of cardinality greater than $\ell_u$.
 \item The weights associated with the same user executing different steps may not increase linearly.
 Once a user has performed one particular unauthorized step, the additional cost of executing a \emph{related} unauthorized step may be reduced, while the additional cost of executing an \emph{unrelated} unauthorized step may be the same as the original cost.
 
 In our purchase order example, we might set $\omega(\set{s_1},u) = M$, where $M$ is a large constant, for any user not authorized for step $s_1$.
 However, we would then set $\omega(\set{s_1,s_3},u) = M$, since the binding-of-duty constraint in $s_1$ and $s_3$ requires the same user performs the two steps.
 \item We can implement separation-of-duty on a per-user basis, which is not possible with UI constraints.  In particular, it may be acceptable for $u_1$ to perform steps $s_1$ and $s_2$, but not $u_2$, in which case $\omega(\{s_1,s_2\},u_1)$ would be small, while $\omega(\{s_1,s_2\},u_2)$ would be large.
 
 In our purchase order example, we would set $\omega(\set{s_4,s_6},u_8) < \omega(\set{s_4,s_6},u_7) < \omega(\set{s_4,s_6},u_1)$, because of the relative seniority of the respective roles to which these users are assigned.
  \item We can minimize the cost of involved users, see Section \ref{sec:mup}.
\end{enumerate}
And in certain cases, it may be that $\omega(T,u) > 0$ for all pairs $(T,u)$.
We see an example of a set-authorization function with this property in Section~\ref{sec:mup}.

\subsection{Optimizing bi-objective WSP}

By analogy to the satisfiability constraint, a weighted constraint $(c,\omega_c)$ is called {\em user-independent}  if, for every permutation $\theta$ of $U$, $\omega_c(\pi)=\omega_c(\theta\circ\pi)$. 
Thus, a weighted UI constraint does not distinguish between users.
Any (weighted) counting constraint, for which the weight of plan $\pi$ is defined in terms of the cardinality of the image of $\pi$, is UI.

Given a plan $\pi:\ S'\rightarrow U$, where $S'\subseteq S$, we define an equivalence relation $\sim_{\pi}$ on $S'$ as follows:  $s \sim_{\pi} s'$ if and only if $\pi(s) = \pi(s')$. This equivalence relation partitions set $S'$ into (non-empty) subsets, called {\em blocks}; the set of blocks is called the  {\em pattern} $P(\pi)$ of $\pi$. We say that two plans $\pi$ and $\pi'$ are {\em equivalent} if they have the same pattern. 
Consider the following plans for our purchase order workflow.
\begin{center}
 \begin{tabular}{|>{$}r<{$}|*{6}{>{$}r<{$}}|}
 \hline
      & s_1 & s_2 & s_3 & s_4 & s_5 & s_6 \\
 \hline
  \tau & u_1 & u_6 & u_1 & u_2 & u_6 & u_5 \\
  \tau' & u_2 & u_6 & u_2 & u_1 & u_6 & u_8 \\
 \hline
 \end{tabular}
\end{center}
Then $\tau$ and $\tau'$ are equivalent as both have pattern $\{\{s_1,s_3\},\{s_4\},\{s_2,s_5\},\{s_6\}\}.$
Observe that if all constraints are UI and plans $\pi$ and $\pi'$ are equivalent, then $\omega_C(\pi)=\omega_C(\pi')$.
Note, however, that $\omega_A(\pi)$ is not necessarily equal to $\omega_A(\pi')$, because $\omega_A(\pi)$ is determined by the set-authorization function, which does depend on the particular users assigned to steps by a plan.
A pattern is said to be {\em complete} if $S'=S$. Plans $\tau$ and $\tau'$ above are complete.

Ideally, we would like to minimize both weights of \BOWSP simultaneously, but, in general, this is impossible as a plan which minimizes one weight will not necessarily minimize the other. 
Thus, we will consider \BOWSP as a bi-objective optimization problem.  
There are three standard approaches to solving bi-objective optimization problems~\cite{Miet}:

\begin{enumerate}
\item We consider a single objective function defined as a linear combination of $\omega_C(\pi)$ and $\omega_A(\pi)$:  that is, we consider an objective function of the form $a\omega_C(\pi)+b\omega_A(\pi)$, where $a$ and $b$ are constants. Since we can incorporate $a$ and $b$ into $\omega_C$ and $\omega_A$, we obtain an optimization problem in which we wish to find a plan $\pi$ of minimum total weight $\omega_C(\pi)+\omega_A(\pi)$. This problem was introduced in \cite{CrGuKa15}, where it was called the {\sc Valued WSP}.  

Henceforth, we will call this problem \BOWSPLC (the LC standing for ``linear combination'').
\item A complete plan $\pi$ is called {\em Pareto optimal} if there is no complete plan $\pi'$ such that 
 \[ \omega_C(\pi)\ge \omega_C(\pi'),\quad \omega_A(\pi)\ge \omega_A(\pi')\quad \text{and}\quad \omega_C(\pi) + \omega_A(\pi)>\omega_C(\pi') + \omega_A(\pi'). \]
In this approach, we minimize one weight subject to the other being upper bounded.
More precisely, given an upper bound $B$ to the constraint (resp. authorization) weight, we aim to find a Pareto optimal plan among all plans $\pi$ such that $\omega_C(\pi) \le B$ (resp. $\omega_A(\pi) \le B$).

Henceforth, we call this problem \BOWSPPO (the PO standing for ``Pareto optimal'').
\item  A pair of complete plans $\pi$ and $\pi'$ are called {\em weight-equal} if 
  \[ \omega_C(\pi)= \omega_C(\pi')\quad \text{and}\quad \omega_A(\pi)=\omega_A(\pi'). \]
In this approach, we compute the {\em Pareto front}, which is a maximal set of non-weight-equal Pareto optimal plans.

Henceforth, we call this problem \BOWSPPF (the PF standing for ``Pareto front'').
\end{enumerate}
%

It is not hard to see that optimal plans obtained by solving \BOWSPLC and \BOWSPPO are Pareto optimal.
Thus, \BOWSPPF generalizes \BOWSPLC and \BOWSPPO. 
It is known that replacing a single objective with multiple objectives (and computing a Pareto front) usually increases the complexity class of the problem \cite{Miet}. 
However, for \BOWSPPF with UI constraints, we have the following theorem.

\begin{thm}\label{thm1} Given a weighted constrained workflow schema $((S,<),U,\omega,C_{\omega})$ in which all constraints are UI, we can solve \BOWSPPF in time\footnote{$O^*(.)$ suppresses any polynomial factors and terms.} $O^*(2^{k\log_2 k})$, where $k$ is the number of steps.
\end{thm}

\begin{proof}
We show how to calculate the Pareto front, which we store in the set $\PF$ and initialize to the empty set.
Recall that for equivalent complete plans $\pi$ and $\pi'$, we have $\omega_C(\pi)=\omega_C(\pi').$
However, $\omega_A(\pi) = \sum_{u\in U} \omega_A(\pi^{-1}(u),u)$ is, in general, different from $\omega_A(\pi')$. 

Let $\pi$ be a complete plan. 
Observe that if there is a complete plan $\pi'$, which is Pareto optimal and equivalent to $\pi$, then $\omega_A(\pi')$ is minimum among all equivalent complete plans equivalent to $\pi$.
Let $P(\pi) = \{T_1,\dots,T_p\}$.
Then to find such a plan $\pi'$ efficiently, we construct a weighted complete bipartite graph $G_{\pi}$ with partite sets $\set{1,\dots,p}$ and $U$ and define the weight of edge $\{q,u\}$ to be $\omega(T_q,u)$.
(Henceforth, for a positive integer $x$, we write $[x]$ to denote $\{1,\dots ,x\}$.) 

Now observe that $G_{\tau}=G_{\pi}$ for every pair $\pi,\tau$ of equivalent complete plans and that $\omega_A(\pi)$ equals the weight of the corresponding matching of $G_{\pi}$ covering all vertices of $[p]$. Hence, it suffices to find such a matching of $G_{\pi}$ of minimum weight, which can be done by the Hungarian method~\cite{Ku55} in time $O(n^3)$ for a graph with $n$ vertices.

Let ${\cal B}_k$ be the number of partitions of the set $[k]$ into non-empty subsets: that is, ${\cal B}_k$ denotes the $k$th Bell number. 
Observe that ${\cal B}_k$ is smaller than $k!$ and there are algorithms of running time $O({\cal B}_k)=O(2^{k\log_2 k})$  to generate all partitions of $[k]$ \cite{Er88}. 
Thus, we can generate all patterns in time $O(2^{k\log_2 k})$. 
For each pattern we compute the corresponding complete plan $\pi$ of minimum authorization weight, and add $\pi$ to $\PF$ unless $\PF$ already has a plan $\tau$ such that 
 \[ \omega_C(\pi)\ge \omega_C(\tau)\quad \text{and}\quad \omega_A(\pi)\ge \omega_A(\tau)\quad \text{and}\quad \omega_C(\pi) + \omega_A(\pi)>\omega_C(\tau) + \omega_A(\tau). \]
The total running time is  $O^*(2^{k\log_2 k})$.
\end{proof}

As we noted above, an algorithm that solves \BOWSPPF can be used to solve \BOWSPLC and \BOWSPPO. 
We thus have the following two corollaries of Theorem~\ref{thm1}:

\begin{cor}
We can compute, in time $O^*(2^{k\log_2 k})$, a plan $\pi$ for \BOWSP with UI constraints such that $\omega_C(\pi) + \omega_A(\pi)$ is minimum.
\end{cor}
\begin{cor}
For any $B \in \mathbb{N}$, we can compute, in time $O^*(2^{k\log_2 k})$, a plan $\pi$ for \BOWSP with UI constraints such that $\omega_A(\pi)$ (resp. $\omega_C(\pi)$) is minimum among all plans $\tau$ such that $\omega_C(\tau) \le B$ (resp. $\omega_A(\tau) \le B$).
\end{cor}

For the sake of completeness, we state the following result whose proof can be found in Appendix~\ref{app:pareto-front-may-have-B-k-points}. 

\begin{thm}\label{thm:tightness}
There exist instances of \BOWSP with UI constraints with a Pareto front consisting of ${\cal B}_k$ points.
\end{thm}

This result implies that the algorithm described in the proof of Theorem~\ref{thm1} has an optimal\footnote{One may possibly improve only the polynomial factors.} asymptotic worst-case running time, as any such algorithm must output all points from the Pareto front (recall that the running time of the algorithm is actually $O^*({\cal B}_k)$).

In addition, note that Gutin and Wahlstr{\"o}m  \cite{GutinW15} recently proved that WSP with UI constraints cannot be solved in time $O^*(2^{ck\log_2 k})$ for any $c<1$, unless the Strong Exponential Time Hypothesis (SETH)\footnote{The Strong Exponential Time Hypothesis claims there is no algorithm of running time $O^*(2^{cn})$ for SAT on $n$ variables, for any $c < 1$~\cite{ImPaZa01}.} fails. Hence, even concerning Approach 1 only (and even with $0$-$1$ weights only), the algorithm in the proof of Theorem~\ref{thm1} has an optimal asymptotic running time possible, under SETH.

\subsection{Minimizing user cost}\label{sec:mup}

The aim of this subsection is to demonstrate that an interesting and applicable problem, the Cardinality-Constrained Minimum User Problem (CMUP)~\cite{Roy2015}, can be reduced to multiple instances of WSP or an instance of BO-WSP.
Moreover, we can solve CMUP in the presence of arbitrary UI constraints.
CMUP asks, given a satisfiable instance of WSP, what is the minimum number of users required to obtain a valid plan. 
Roy {\em et al.} \cite{Roy2015} motivate their study of CMUP by noting that many companies and other organizations are interested in cost-cutting in their businesses. 
However, CMUP can be useful in other scenarios, such as when a smaller number of users increases security.  

For CMUP, Roy {\em et al.} \cite{Roy2015} used only constraints of the form at-most-$r$ and at-least-$r'$ and suggested using either an Integer Programming Solver (they used CPLEX) or a bespoke greedy-like heuristic generalizing a greedy-like heuristic for coloring hypergraphs. 

Note, however, that one can solve CMUP using any WSP solver that can deal with at-most-$r$ constraints, which are UI for any integer $r$.
In particular, the WSP solver developed by Cohen {\em et al.}~\cite{CoCrGaGuJo14c,KaGaGu} for solving WSP with UI constraints can be used. 
First note that no plan for a satisfiable instance requires more than $k$ users.
Hence, we can use binary search, with repeated calls to the WSP solver, to solve CMUP.
Before making the first call to the WSP solver, we add the constraint $\atmost(S,k/2)$ to $W$, and test whether the resulting instance is satisfiable or not. 
If it is satisfiable, we then replace $\atmost(S,k/2)$ with $\atmost(S,k/4)$ and call the WSP solver; otherwise we replace $\atmost(S,k/2)$ with $\atmost(S,3k/4)$ (and call the WSP solver).
By iterating this process at most $\lceil \log_2 k \rceil$ times, we obtain the minimum number of users (recall that the CMUP instance is required to be satisfiable).

In contrast, the \BOWSP algorithm can solve CMUP directly.  
In fact, \BOWSP allows us to model and solve the following generalization of CMUP.
We will allow any weighted UI constraints. Assume that the users' costs vary (salaries, costs of additional training, etc.) and let $\mu_u > 0$ be the cost of user $u$. Then we can define a weighted set-authorization function as follows:
 \[
 \omega(T,u)=
  \begin{cases}
   \mu_u & \text{if $(u,t) \in A$ for all $t \in T$}, \\
   M & \text{otherwise},
  \end{cases}
\]
where $M$ is a constant such that $M>\sum_{u\in U}\mu_u$. 
(Note here that $\omega(T,u) > 0$ for all $T \ne \emptyset$.)
By definition of \BOWSP, the algorithm will find a valid plan involving a set of users of minimum cost (the CMUP case corresponds to $\mu_u = 1$ for every user $u$). We may also consider decreasing the cost by allowing some ``minor'' constraints to be violated.

\section{\BOWSP algorithms and experimental results}
\label{sec:experiments}

In this section, we describe two concrete algorithms we have developed to solve \BOWSP and study their performance on a range of \BOWSP instances.
The first, which we call \emph{Pattern Branch and Bound}, solves \BOWSPPF directly; the algorithm is FPT and based on patterns.
The second algorithm, which we call \emph{MIP-based}, solves \BOWSPPF by solving multiple instances of \BOWSPPO expressed as a mixed integer programming (MIP) problem.

Note that a Pareto front of \BOWSP may include plans with very large authorization or constraint weights despite the fact that such plans are of little practical interest. 
For that reason we introduce two parameters, $B_A$ and $B_C$, which bound the authorization and constraint weights, respectively.  
Both of our solution methods ignore plans $\pi$ such that $\omega_A(\pi) > B_A$ or $\omega_C(\pi) > B_C$. 
These upper bounds not only reduce the size of the Pareto front, thereby reducing the number of choices that an end-user would need to evaluate, but, as we will show later, also speed up the algorithms.
 
\subsection{Pattern backtracking algorithm}
 
 Pattern Branch and Bound (PBB) is an efficient implementation of the algorithm described in the proof of Theorem~\ref{thm1}.
 The algorithm explores the space of patterns in a depth-first search like manner and, for each complete pattern, computes an optimal complete plan.
 The algorithm maintains a set $\PF$ of Pareto optimal plans found so far.
 When the search is completed, $\PF$ is the full Pareto front.
 
 The search tree of the algorithm is formed as follows.
 Each node of the search tree corresponds to a (partial) pattern.
 Branching in PBB is based on extending the parent pattern with one step.
 Let $P = \set{T_1, T_2, \ldots, T_{|P|}}$ be the (partial) pattern in some node of the search.
 Then PBB will select some step $s \in S \setminus (T_1 \cup T_2 \cup \ldots \cup T_{|P|})$ and extend $P$ with $s$ in every possible way, i.e.
 \begin{align*}
 & \set{T_1 \cup \set{s},\ T_2,\ \ldots,\ T_{|P|}}, \\
 & \set{T_1,\ T_2 \cup \set{s},\ \ldots,\ T_{|P|}}, \\
 & \ldots \\
 & \set{T_1,\ T_2,\ \ldots,\ T_{|P|} \cup \set{s}}, \\
 & \set{T_1,\ T_2,\ \ldots,\ T_{|P|},\ \set{s}}.
 \end{align*}
 
 Like most tree-search algorithms, PBB uses pruning when it can guarantee that the current branch of the search contains no new Pareto optimal solutions.
 Since all weights are non-negative, extending a pattern can only increase the weight of a pattern.
 Thus, if we can establish that $\PF$ already includes a plan that is weight-equal to or dominated by every plan in the current branch of the search, we may terminate the search of the sub-tree rooted at $P$.
 More specifically, for a given pattern $P$ defined over steps $T \subseteq S$, we define
 \[
  \underline{\omega_c}(P) \stackrel{\rm def}{=} \min\set{\omega_c(\pi) : \text{$\pi$ is a complete plan such that $P(\pi|_T) = P$}}.
 \]
 For example, suppose $c$ is an at-least-3 constraint with scope $T$, where $|T| = 5$, and four of the steps in $T$ are already assigned the same user (by pattern $P$).
 Then $\underline{\omega_c}(P)$ is equal to the penalty for assigning two users to the steps in $T$ (since we may assume that this penalty will be no greater than the penalty for assigning a single user to all five steps in $T$).
 We then compute lower bounds for $P$:
 \[
  \underline{\omega_A}(P) \stackrel{\rm def}{=} \sum_{B \in P} \min_{u \in U} \omega(B, u)\quad\text{and}\quad 
  \underline{\omega_C}(P) \stackrel{\rm def}{=} \sum_{c \in C} \underline{\omega_c}(P).
 \]
 Then we may prune the branch rooted at $P$ if one of the following conditions holds:
  \begin{itemize}
   \item there exists a plan $\pi'$ in $\PF$ such that $\omega_A(\pi') \le \underline{\omega_A}(P)$ and $\omega_C(\pi') \le \underline{\omega_C}(P)$;
   \item $\underline{\omega_A}(P) > B_A$;
   \item $\underline{\omega_C}(P) > B_C$.
  \end{itemize}
 
 The size of the search tree depends on the order of steps by which the pattern is extended.
 This order is determined dynamically by a heuristic which selects the constrained step $s$ most likely to reduce the average branching factor.
 The implementation of the heuristic depends on particular types of constraints and users, and in our case was adjusted by automated parameter tuning, as we describe in more detail in Appendix~\ref{sec:bh}.

\subsection{MIP-based algorithm}

 Modern, general-purpose MIP solvers are convenient and powerful tools that can be used to tackle many hard combinatorial optimization problems, thereby reducing costs resulting from the development and maintenance of bespoke systems.
 In this section we show how an MIP solver can be used to solve \BOWSPPF.

 Current MIP solvers can only return a single solution in a run.
 \BOWSPPF, however, is a bi-objective problem with multiple Pareto optimal solutions in general.
 Thus, to solve \BOWSPPF, we need a control algorithm that calls an MIP solver multiple times with different parameters. 
 We now describe the construction of such an algorithm in more detail.

\subsubsection{Computing a Pareto Front}

 Given an instance of \BOWSP and integers $a$, $b$, $c$ and $d$, suppose we have a function {\it mipMinimizeWeight} that can be evaluated using an MIP solver and returns a plan $\pi$ such that
  \begin{inparaenum}[(i)]
   \item either the authorization weight or constraint weight of $\pi$ is minimized, 
   \item $\omega_A(\pi) \in [a,b]$, and
   \item $\omega_C(\pi) \in [c,d]$.
  \end{inparaenum}
 We explain how to construct the required MIP formulation for {\it mipMinimizeWeight} in Section~\ref{sec:mipMinimizeWeight}.
 
 Let us assume that $\mathit{mipMinimizeWeight}(0,a,b,c,d)$ returns a plan $\pi$ such that the \emph{authorization} weight is minimized, and $\mathit{mipMinimizeWeight}(1,a,b,c,d)$  returns a plan such that the \emph{constraint} weight is minimized.
 Then we can compute the Pareto optimal plan $\pi^{\ell}$ such that $\omega_A(\pi^{\ell}) \le \omega_A(\tau)$ for any Pareto optimal plan $\tau$ using the following computations:
 \begin{align}
	\pi &\leftarrow \mathit{mipMinimizeWeight}(0,0,B_A,0,B_C) \label{eqn:initpi} \\
	\pi^{\ell} &\leftarrow \mathit{mipMinimizeWeight}(1,\omega_A(\pi),\omega_A(\pi),0,B_C). \label{eqn:pileft}
 \end{align}
 The method is illustrated schematically in Figure~\ref{fig:computing-pareto-optimal-plan-minimizing-authorization}.
 (The black dot in each sub-figure indicates the optimal plan, subject to the particular weight being minimized, which is indicated by a black arrow.
 The gray area in each sub-figure indicates the region within which we are seeking a plan.)
 
\begin{figure}[!htb]\centering
%
\begin{subfigure}[t]{.45\textwidth}
	\centering
	\includegraphics{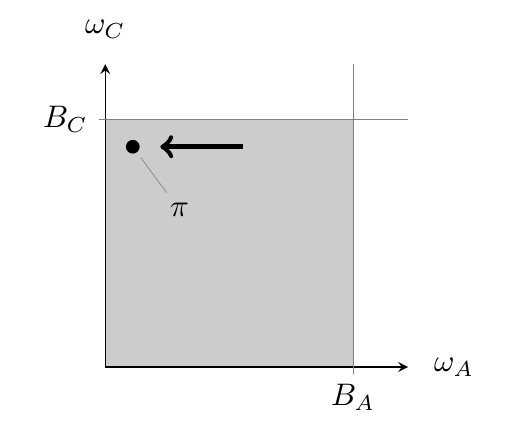}

%
%
%
        \caption{Step 1 (see (\ref{eqn:initpi}))}
\end{subfigure}
\hfill
\begin{subfigure}[t]{.45\textwidth}
	\centering
	\includegraphics{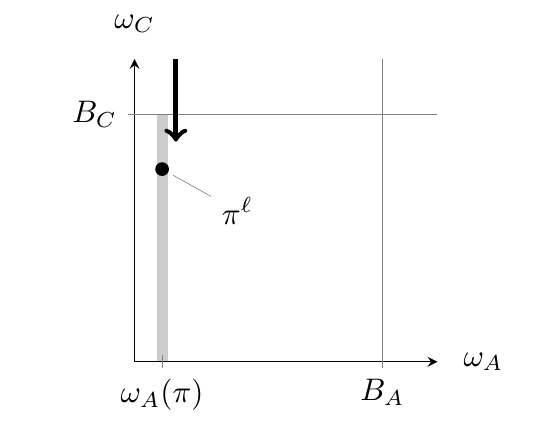}

%
%
%
%
	\caption{Step 2 (see (\ref{eqn:pileft}))}
\end{subfigure}
\caption{Computing a Pareto optimal plan having minimum authorization weight}
\label{fig:computing-pareto-optimal-plan-minimizing-authorization}
\end{figure}

 \noindent%
 Similarly, computing 
 \begin{align}
  \pi &\leftarrow \mathit{mipMinimizeWeight}(1,0,B_A,0,B_C) \\
  \pi^r &\leftarrow \mathit{mipMinimizeWeight}(0,0,B_A,\omega_C(\pi),\omega_C(\pi))
 \end{align}
 will compute a Pareto optimal plan $\pi^r$ having the smallest constraint weight (among all Pareto optimal plans).

 These steps form the basis for computing the entire Pareto front, as illustrated by the pseudo-code in Algorithm~\ref{alg:mip-based}.
 Informally, we compute the ``extremal'' Pareto optimal plans $\pi^{\ell}$ (lines \ref{line:start-compute-left-pareto-optimal}--\ref{line:end-compute-left-pareto-optimal}) and $\pi^r$ (lines \ref{line:start-compute-right-pareto-optimal}--\ref{line:end-compute-right-pareto-optimal}).%
   \footnote{In line~\ref{line:start-compute-right-pareto-optimal}, we wish to find plan $\pi^r$ such that $\omega_C(\pi^r) < \omega_C(\pi^{\ell})$.  
	     However, MIP does not allow strict inequality constraints; thus we must use a constraint $\omega_C(\pi^r) \le \omega_C(\pi^{\ell}) - \epsilon$ for a small $\epsilon > 0$.  
	     Since all the weights in our experiments are integers, we can set $\epsilon = 1$ without loss of correctness.
	     Analogous comments hold for $\omega_C(\pi^r)$ and $\omega_A(\pi^r)$.} 
 Having done that, we iteratively shrink the (two-dimensional) objective space, moving the bottom-right hand corner towards the top left-hand corner, and find the remaining Pareto optimal plans within those successively smaller regions (lines~\ref{line:start-compute-mid-pareto-optimal}--\ref{line:end-compute-mid-pareto-optimal}).
 The operation of the algorithm is illustrated schematically in Figure~\ref{fig:mip-based}.
 
\begin{algorithm2e}[!h]
\caption{MIP-based algorithm to compute a Pareto front}
\label{alg:mip-based}

$\pi \gets \mathit{mipMinimizeWeight}(0,0,B_A,0,B_C)$\;\label{line:start-compute-left-pareto-optimal}
$\pi^{\ell} \gets \mathit{mipMinimizeWeight}(1,\omega_A(\pi),\omega_A(\pi),0,B_C)$\;\label{line:end-compute-left-pareto-optimal}
\If {$\pi^{\ell} = \mathit{(none)}$}
{
	\Return $\emptyset$\;
}

$\PF \gets \set{\pi^{\ell}}$\;

$\pi \gets \mathit{mipMinimizeWeight}(1,\omega_A(\pi^{\ell}) + 1,B_A,0,\omega_C(\pi^{\ell})-1)$\;\label{line:start-compute-right-pareto-optimal}
$\pi^r \gets \mathit{mipMinimizeWeight}(0,\omega_A(\pi^{\ell}) + 1,B_A,\omega_C(\pi),\omega_C(\pi))$\;\label{line:end-compute-right-pareto-optimal}
\If {$\pi^r = \mathit{(none)}$}
{
	\Return $\PF$\;
}

$\PF \gets \PF \cup \set{\pi^r}$\;

\While {$\omega_A(\pi^r) - \omega_A(\pi^{\ell}) > 1$ and $\omega_C(\pi^{\ell}) - \omega_C(\pi^r) > 1$} { \label{line:begin-loop}
	$\pi \gets \mathit{mipMinimizeWeight}(1,\omega_A(\pi^{\ell}) + 1,\omega_A(\pi^r)-1,\omega_C(\pi^r)+1,\omega_C(\pi^{\ell})-1)$\;\label{line:start-compute-mid-pareto-optimal}
	$\pi^c \gets \mathit{mipMinimizeWeight}(0,\omega_A(\pi^{\ell}) + 1,\omega_A(\pi^r)-1,\omega_C(\pi),\omega_C(\pi))$\;\label{line:end-compute-mid-pareto-optimal}
	\If {$\pi^c = \mathit{(none)}$}
	{
		\Return $\PF$\;
	}
	
	$\PF \gets \PF \cup \set{\pi^c}$\;
	$\pi^r \gets \pi^c$\;\label{line:end-loop}
}

\Return $\PF$\;
\end{algorithm2e}

\begin{figure}[!htb]\centering

\begin{subfigure}[t]{.3\textwidth}
	\centering
		\includegraphics{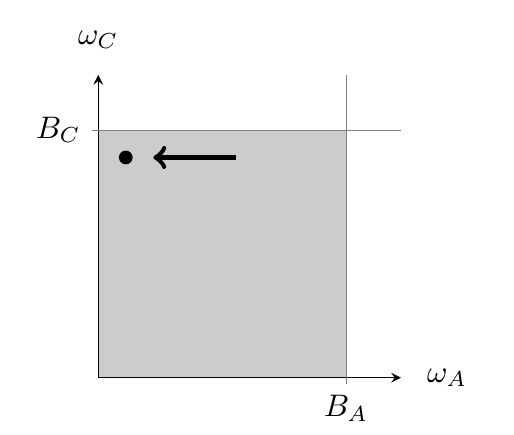}
        \caption{}
	\label{fig:mip1}
\end{subfigure}
\hfill
\begin{subfigure}[t]{.3\textwidth}
	\centering
	\includegraphics{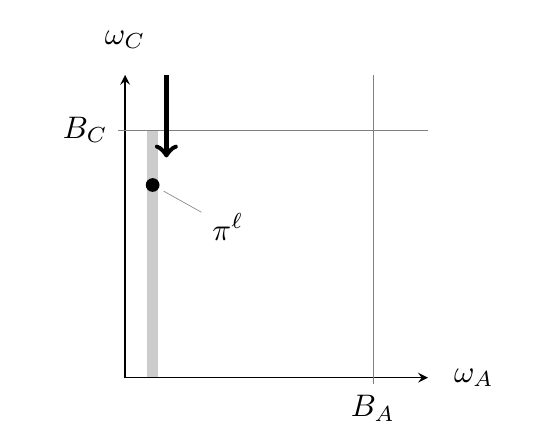}
	\caption{}
	\label{fig:mip2}
\end{subfigure}
\hfill
\begin{subfigure}[t]{.3\textwidth}
	\centering
		\includegraphics{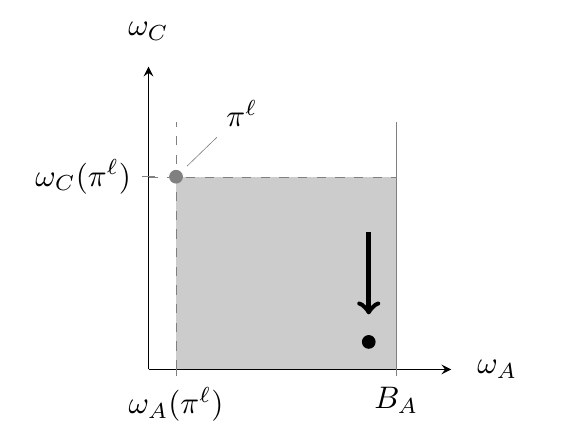}
	\caption{}
	\label{fig:mip3}
\end{subfigure}
\\
\begin{subfigure}[t]{.3\textwidth}
	\centering
		\includegraphics{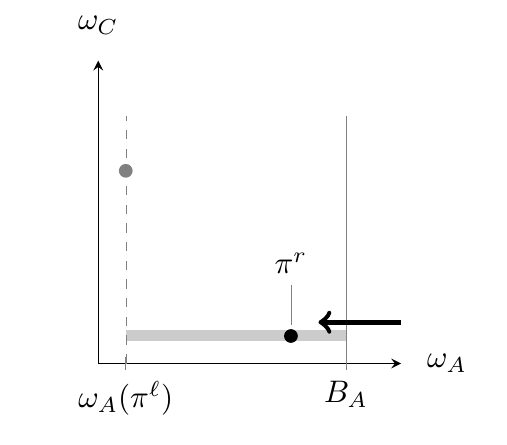}
	\caption{}
	\label{fig:mip4}
\end{subfigure}
\hfill
\begin{subfigure}[t]{.3\textwidth}
	\centering
	\includegraphics{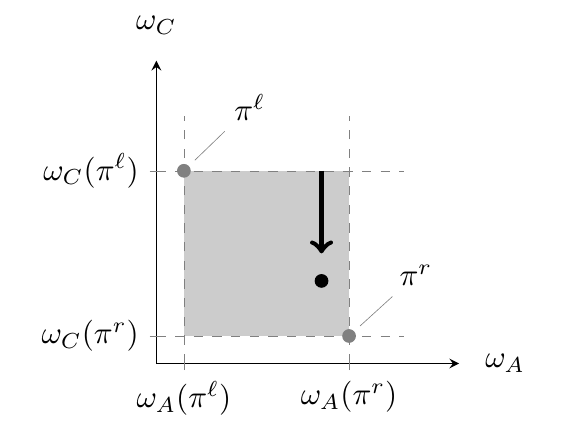}
	\caption{}
\end{subfigure}
\hfill
\begin{subfigure}[t]{.3\textwidth}
	\centering
		\includegraphics{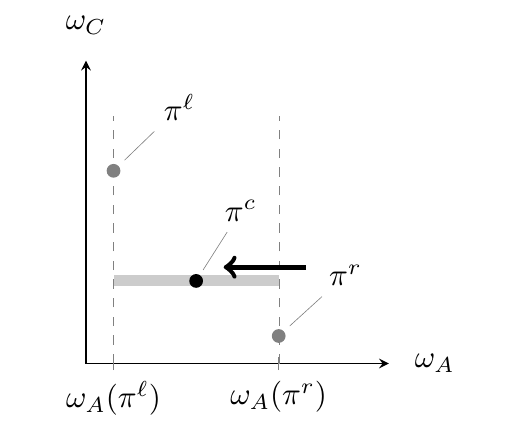}
	\caption{}
	\label{fig:mip6}
\end{subfigure}
\caption{
	Schematic representation of the computation of the Pareto front using our MIP-based algorithm}
\label{fig:mip-based}
\end{figure}

\subsubsection{The function {\it mipMinimizeWeight}}\label{sec:mipMinimizeWeight}

 In order to formulate $\mathit{mipMinimizeWeight}$ as an MIP problem, we define a matrix of binary decision variables $x = (x_{su})$ to represent a plan: specifically, if $x_{su} = 1$ then step $s \in S$ is assigned to user $u \in U$; otherwise step $s$ is assigned to some other user.
 Then the constraint weight $\omega_C(x)$ and the authorization weight $\omega_A(x)$ of the plan defined by $x$ are
 \[
  \omega_A(x) \stackrel{\rm def}{=} \sum_{u \in U} \omega_u(x) \quad\text{and}\quad \omega_C(c) \stackrel{\rm def}{=} \sum_{c \in C} \omega_c(x),
 \]
 where the encodings of $\omega_u(x)$ and $\omega_c(x)$ depend on user and constraint types (and may require the introduction of additional decision variables).%
  \footnote{For example, a separation-of-duty constraint $c$ with scope $\{s, t\}$ requires the introduction of a binary decision variable $p_c$, where $p_c \ge x_{su} + x_{tu} - 1$.
	    Thus, $p_c = 0$ if and only if $s$ and $t$ are assigned to different users.
	    If the violation of $c$ is associated with weight $x$, say, then we define $\omega_c(x) = xp_c$.}
 We describe the encoding of $\omega_u(x)$ and $\omega_c(x)$ for the users and constraints used in our experiments in Appendix~\ref{sec:mip-encodings}.
 
 Having defined the functions $\omega_A(x)$ and $\omega_C(x)$, we formulate the MIP problem as follows:
 \begin{align}
  \label{eq:mip-objective}
   & \text{minimize } (1 - \alpha) \omega_A(x) + \alpha \omega_C(x) \\
   & \text{subject to}
   \label{eq:omega_A_bounds}
     && a \le \omega_A(x) \le b, \\
   \label{eq:omega_C_bounds}
     &&& c \le \omega_C(x) \le d, \\
   \label{eq:mip-x}
     &&& \sum_{u \in U} x_{su} = 1 \qquad \forall s \in S, \\
  \label{eq:x-values}
    &&& x_{su} \in \set{0, 1} \qquad \forall s \in S, \forall u \in U.
 \end{align}
 The parameter $\alpha$ determines the relative importance of the authorization and constraint weights.
 In particular, if $\alpha = 0$ then the problem is to minimize the authorization weight only; and if $\alpha = 1$ then it is to minimize the constraint weight only.
 Constraint~\eqref{eq:x-values} ensures that each variable $x_{su}$ takes binary values.
 Constraints~\eqref{eq:omega_A_bounds} and~\eqref{eq:omega_C_bounds} restrict the search space to the bounds defined by $a$, $b$, $c$ and $d$.
 Finally, constraint~\eqref{eq:mip-x} ensures that each step is assigned to exactly one user.
 The solution $x$ can be trivially translated into the plan $\pi$ to be returned by $\mathit{mipMinimizeWeight}$: $\pi(s) = u$ if and only if $x_{su} = 1$.

\subsection{Testbed}
\label{sec:testbed}

 We use a pseudo-random instance generator to produce benchmark instances.
 Our generator is an extension of generators for WSP~\cite{KaGaGu} and Valued WSP~\cite{CrGuKa15}.
 
 For ease of reference, we define the \emph{authorization density} of a WSP instance to be the average
number of permitted steps per user, where a step is said permitted if its
corresponding cost is not prohibitive, \textit{i.e.} is smaller than a given constant $M$ (see later).
 Similarly, we define the \emph{separation-of-duty constraint density} to be the ratio of the number of separation-of-duty constraints to $\binom{k}{2}$, since there are at most $\binom{k}{2}$ choices for pairs of steps in a separation-of-duty constraint.
 
 Our \BOWSP instance generator takes the following parameters: (i) number of steps $k$, (ii) authorization density $d$, (iii) separation-of-duty constraint density $e$, and (iv) seed value for the pseudo-random number generator.
  All the constraints are generated randomly.
 In particular, the scopes are produced randomly and uniformly such that no two constraints of the same type have identical scopes; penalties (i.e., weights) are also selected randomly and uniformly, with allowed ranges given in Figure~\ref{fig:constraint-penalties}. Our choice of penalties is based on the fact that we view at-most-3 constraints as soft constraints (i.e., constraints that a plan may violate, if needed) and we wish the penalty of assigning four users to the scope of such a constraint to be smaller than assigning five. We view at-least-3 constraints as ``semi-soft'', i.e., we issue a small penalty for assigning just two users, but basically forbid assigning just one (so that we forbid ``dictators''). We view separation-of-duty constraints as hard constraints, i.e. constraints that a plan must satisfy. 
The values of parameters and penalty functions used in our generator are obtained empirically such that the instances are relatively hard to solve.
The generator produces $\left\lfloor e \binom{k}{2} + \frac{1}{2} \right\rfloor$ separation-of-duty constraints, $k$ at-most-3 constraints and $k$ at-least-3 constraints.

\begin{figure}[htb]
\begin{subfigure}[b]{.3\textwidth}\centering
	\begin{tabular}{@{} cr @{}}
		\toprule
 		Distinct users & Penalty \\
		\midrule
		1 & 0 \\
		2 & 0 \\
		3 & 0 \\
		4 & 3--5 \\
		5 & 10--15 \\
		\bottomrule
	\end{tabular}
\caption{At-most} 
\end{subfigure}
\hfill
\begin{subfigure}[b]{.3\textwidth}\centering
	\begin{tabular}{@{} cr @{}}
		\toprule
 		Distinct users & Penalty \\
		\midrule
		1 & 1,000,000 \\
		2 & 1--3 \\
		3 & 0 \\
		4 & 0 \\
		5 & 0 \\
		\bottomrule
	\end{tabular}
\caption{At-least}
\end{subfigure}
\hfill
\begin{subfigure}[b]{.3\textwidth}\centering
	\begin{tabular}{@{} cr @{}}
		\toprule
 		Distinct users & Penalty \\
		\midrule
		1 & 1,000,000 \\
		2 & 0 \\
		\\
		\\
		\\
		\bottomrule
	\end{tabular}
 \caption{Separation-of-duty constraints} 
\end{subfigure}
 \caption{Penalties associated with constraints produced by our \BOWSP instance generator; where a range is given, the penalty is selected from the range uniformly at random}\label{fig:constraint-penalties}
\end{figure}

 The generator creates $10k + 10$ users, based on the assumption that the number of users will be an order of magnitude greater than the number of steps.
 In particular, it produces $10k$ \emph{staff members} and 10 \emph{(external) consultants}.
 The weight function for these two types of users are different.
 In particular, consultants are not authorized for any steps, so there is always a non-zero penalty associated with any plan that assigns a consultant to a step.
 
 We assume that each member of staff $u$ is authorized for some subset $A_u$ of $S$ and, in the absence of authorized users, could be assigned to some subset $B_u$ of $S$ (thereby incurring a modest penalty).
 Thus, each staff member $u$ is assigned two non-intersecting sets of steps, $A_u, B_u \subset S$, where $|A_u|$ has Poisson distribution\footnote{We artificially cap the size $|A_u|$ by $k - 2$.} with $\lambda = d$, and $|B_u| = 2$.
 Then the penalty function $\omega(T, u)$ for a staff member $u$ is defined to be 
  \[ \omega(T, u) = \sigma_u \cdot |T \cap B_u| + M \cdot |T \setminus (A_u \cup B_u)|, \] 
 where $\sigma_u \in [5, 15]$ is selected for each member of staff randomly and uniformly, and $M$ is a large constant (1,000,000 in our implementation).
 
 Similarly, each consultant $u$ is assigned a set $B_u$ such that $|B_u|$ has Poisson distribution (capped artificially by $k$) with $\lambda = d$.
 This set represents steps that a consultant could perform in the absence of authorized users without incurring too large a penalty.
 Then the penalty function for consultants is defined to be
 $$
 \omega(T, u) = 
 	\begin{cases}
 		0 & \text{if } T = \emptyset, \\
 		\sigma_u & \text{if } T \subseteq B_u, \\
 		M & \text{otherwise,}
	\end{cases}
 $$
 where $\sigma_u \in [10, 30]$ is selected randomly and uniformly.
 
 The MIP encodings of the (weighted) constraints used in our experiments are described in Appendix~\ref{sec:mip-encodings}. 
 The source code of our instance generator can be found at~\cite{sourceSACMAT2015}. 

\subsection{Computational experiments}
\label{sec:experimentsresults}

 In this section we report the results of our computational experiments with \BOWSP.
 Our algorithms are implemented in C\#, and the MIP solver used by the MIP-based algorithm is CPLEX~12.6.
 Our test machine is based on two Intel Xeon CPU E5-2630 v2 (2.6~GHz) and has 32~GB RAM installed.
 Hyper-threading is enabled, but we never run more than one experiment per physical CPU core concurrently, and concurrency is not exploited in any of the tested solution methods.
 
 In Figure~\ref{fig:pareto-fronts}, we show examples of optimal solutions (full Pareto fronts) of four \BOWSP instances with various authorization densities $d$ and separation-of-duty constraint densities $e$.
 The objective upper bounds $B_A$ and $B_C$ are set to 1000 in all four cases.
 The relatively lightly constrained instance ($d = 20\%$ and $e = 10\%$) is almost satisfiable; it has only two Pareto optimal solutions $(0, 1)$ and $(5, 0)$, where each tuple gives the objectives $(\omega_A(\pi), \omega_C(\pi))$.
 All other instances used in this experiment do not have plans that would satisfy all the constraints (at least for $\omega_A(\pi) \le B_A$), but they all have plans that satisfy all the authorizations.
 As we noted earlier, given at least one user is authorized for each step $s \in S$, there necessarily exists a plan $\pi$ such that $\omega_A(\pi) = 0$; $\omega_C(\pi)$, however, may well be very large.
%

\begin{figure}[h]\centering
	\includegraphics{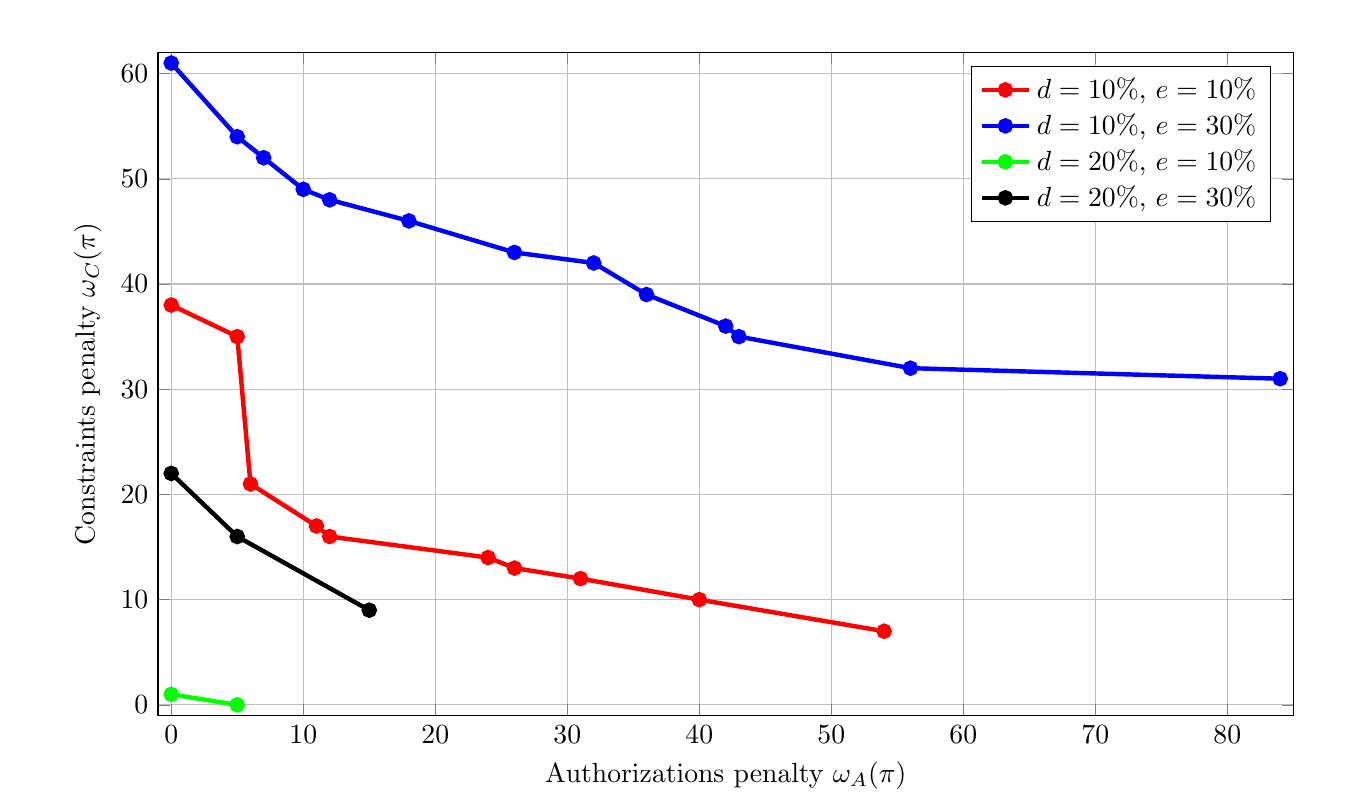}

%
%
%
\caption{Optimal solutions for four \BOWSP instances}\label{fig:pareto-fronts}
\end{figure}

 In Figure~\ref{fig:scaling-k}, we compare the performance of the PBB (solid lines) and MIP-based algorithms (dashed lines) as the number $k$ of steps grows.
 Each experiment is repeated 100 times (on 100 different instances obtained by varying the pseudo-random number generator seed value); median running times have been reported; and $B_A = B_C = 1000$.	
 PBB shows a slightly super-exponential growth, consistent with the worst case time complexity $O^*({\cal B}_k) = O^*(2^{k \log_2 k})$.
 The scaling factor is roughly the same for all the classes of instances, but the constant factor is higher for the more constrained instances.
 
 The MIP-based algorithm is significantly slower than PBB on most of the instances, but shows better scaling than PBB on instances with low authorization density.
 This suggests that the branching heuristic of PBB will need to be improved in the future to better account for authorizations.
 On the other hand, PBB is an FPT algorithm and, thus, its running time scales \emph{polynomially} with the number of users $n$ (see Section~\ref{sec:fpt-results}), whereas the MIP-based algorithm scales \emph{exponentially} with $n$.
 Indeed, our MIP formulation does not exploit the FPT nature of \BOWSP, and it is very unlikely that an MIP solver will be able to discover the FPT nature of the problem itself.
 The exponential scaling of the MIP-based algorithm's running time was also confirmed empirically (the results are not reported here).
 As a result, the MIP-based algorithm is particularly weak when $n$ is large relative to $k$.
 
\begin{figure}[h]\centering
	\includegraphics{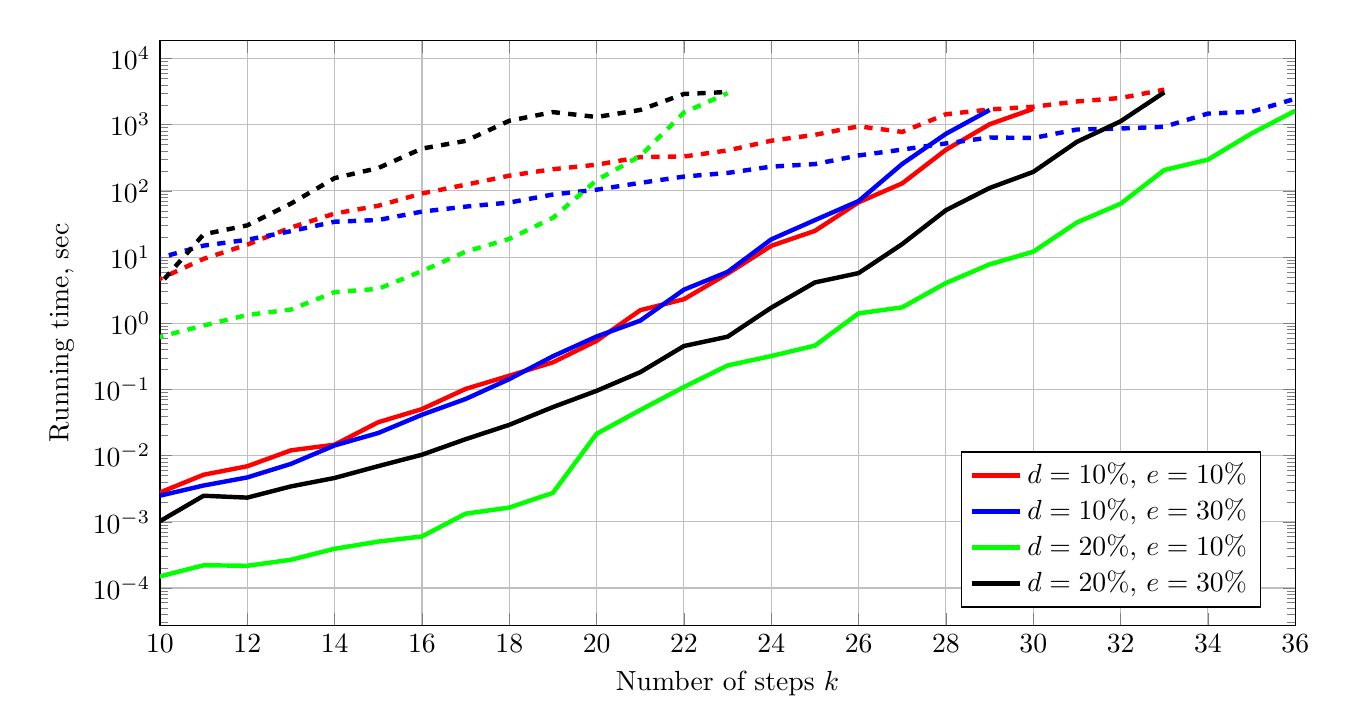}
	
%
%
%
\caption{A comparison of running times for the PBB and the MIP-based algorithms. We use dashed lines for the MIP algorithm and solid for PBB.}
\label{fig:scaling-k}
\end{figure}

 We have seen in Figure~\ref{fig:pareto-fronts} that most of the Pareto-optimal plans in our instances have moderate authorization and constraint weights.
 Nevertheless, parameters $B_A$ and $B_C$ play an important role in the performance of our algorithms.
 Figure~\ref{fig:scaling-b} shows how the running times of PBB and the MIP-based algorithm depend on $B_A$ and $B_C$. Here we are solving instances with parameters $k = 20$, $d = 10\%$ and $e = 30\%$.
 Each experiment is repeated 100 time (on 100 different instances) and the median running time is reported.
 It is easy to see that both approaches are sensitive to the values of $B_A$ and $B_C$, but PBB is particularly good at exploiting tight upper bounds.
 This suggests that setting realistic values for $B_A$ and $B_C$ can significantly improve the performance of PBB (recall that in all the previous experiments we used relatively large values $B_A = B_C = 1000$).

\begin{figure}[h]\centering
	\includegraphics{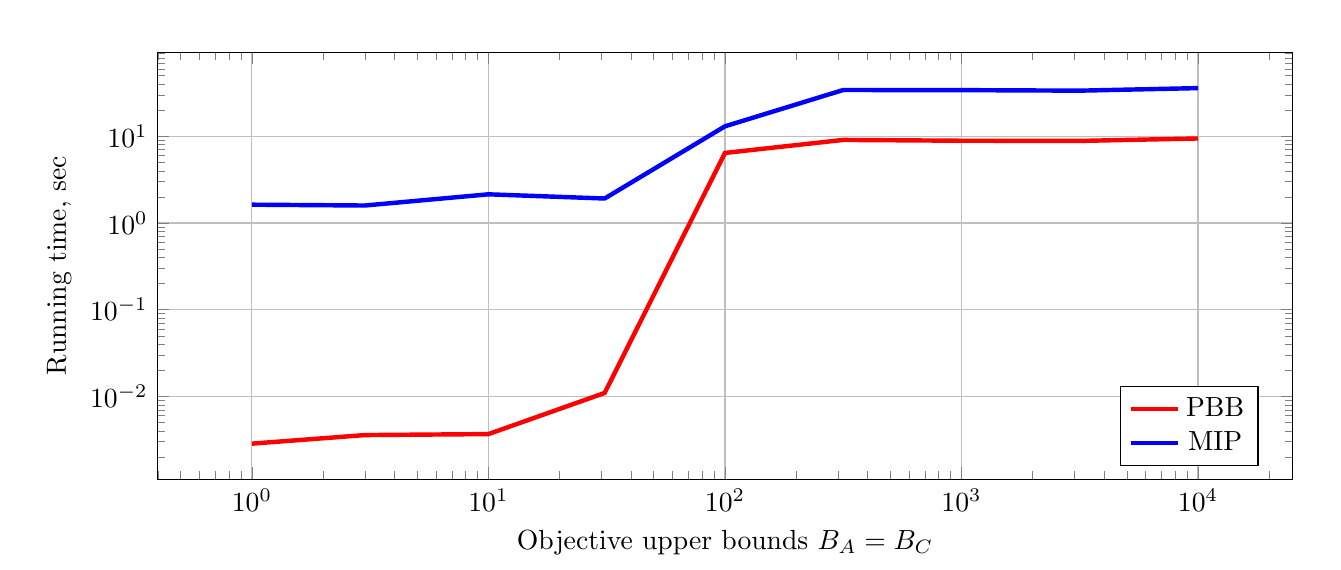}
	
%
\caption{Dependence of the algorithms' performance on the objective upper bounds $B_A$ and $B_C$}\label{fig:scaling-b}
\end{figure}

\section{Resiliency}\label{sec:resiliency}

If we assume that all users are available at all times and that all constraints are UI, then the authorization relation $A$ is fixed for the duration of a workflow instance.
Moreover, the construction of a plan is independent of the order in which the steps are executed in any workflow instance.  
However, in practice, users may be unavailable for a wide variety of reasons, which means, essentially, that the authorization relation varies over time.
If a user $u$ has been assigned to a step $s$ by a plan $\pi$ for a workflow instance $I$ and that user is absent, then $I$ cannot complete until $u$ becomes available again.
Similarly, if users are not pre-assigned to steps in a workflow instance $J$ but the only users authorized for the steps that are ready to be performed in $J$ are unavailable, then the execution of the workflow instance is suspended until one or more of those users becomes available.

Resiliency is the study of workflow satisfiability under the assumption that the user population may vary over time.
It was first studied by Wang and Li~\cite{WaLi10} and a number of results are known about the computational complexity of decision problems associated with workflow resiliency.
In this section, we establish the first FPT results for these decision problems.
Then, acknowledging the argument that knowing the answer to these decision problems is unlikely to be of practical use~\cite{MaMoMo14}, we introduce a new model for resiliency and relate problems associated with this model to \BOWSP.

Before proceeding further, it is useful to introduce the notions of order ideals and linear extensions in a partially ordered set.
An \emph{order ideal} $Y$ in a partially ordered set $(X,<)$ is a subset of $X$ such that for any $y \in Y$ and any $x \in X$, $x < y$ implies $x \in Y$.  
A partially complete instance of a workflow specification $W = ((S,<),U,A,C)$ may be modeled as a pair $(W,\pi)$, where $\pi$ is a partial plan $\pi : T \rightarrow U$ for some order ideal $T$ in the partial order defined by $(S,<)$.  
The set $T$ models the steps that have been executed in the workflow instance. 

The set of steps $\{x \in S : \{x\} \cup T \text{ is an order ideal in } S\}$ denotes the set of \emph{ready} steps.  
Thus a complete instance defines a sequence of order ideals $[T_1, \dots, T_k]$, where $|T_1| = 1$, $T_{i} \subseteq T_{i+1}$ and $|T_{i+1}| = |T_i|+1$ for all $i \in \{0, \dots, k-1\}$; equivalently, a complete instance defines a sequence of steps $[s_1,\dots,s_k]$, where $\{s_i\} = T_i \setminus T_{i-1}$ (with $T_0$ being assumed to be $\emptyset$) for $i \in [k]$. 
It is easy to show that for any such sequence we have $s_i < s_j$ for all $i \le j$ (as the contrary would mean that a non-ready step had been executed during the workflow instance).  
A sequence of steps with this property is said to be a \emph{linear extension} of $(S,<)$~\cite{DaPr02}. 
Every linear extension represents a possible order of execution of the steps in $S$.
The linear extensions for the workflow depicted in Figure~\ref{fig:example-workflow}, for example, are $[ s_1, s_2, s_3, s_4, s_5, s_6 ]$, $[s_1, s_2, s_3, s_5, s_4, s_6]$, and $[s_1, s_2, s_4, s_3, s_5, s_6]$.

\subsection{The work of Wang and Li}

Wang and Li introduced the notion of workflow resiliency to reason about situations in which one or more users might not be available to execute steps in a workflow~\cite{WaLi10}.
Essentially, workflow resiliency asks whether a workflow remains satisfiable even if some users are not available.
Wang and Li considered three different scenarios of user availability over the lifetime of a workflow instance:
 \begin{itemize}
  \item \emph{static} resiliency is concerned with the scenario that some fixed subset of users is unavailable;
  \item \emph{decremental} resiliency is concerned with the scenario that the set of available users decreases monotonically; and
  \item \emph{dynamic} resiliency is concerned with the scenario that an arbitrary set of users is unavailable at any given point in time.
 \end{itemize}
More formally, we have the following definitions~\cite{WaLi10}.
In Section~\ref{sec:res:wangli}, we present equivalent and, arguably, more intuitive definitions in order to prove our FPT results.
 
\begin{df}
Given an integer $t \ge 0$, a constrained workflow authorization schema $((S, <), U, A, C)$ is {\em statically $t$-resilient} if, for all $U' \subseteq U$ such that $|U'| \ge n-t$, $((S, <), U', A, C)$ is satisfiable.
\end{df} 

 
\begin{df} \label{def:decrres}
Given an integer $t \ge 0$, a constrained workflow authorization schema $((S, <), U, A, C)$ is {\em decrementally $t$-resilient} if and only if Player 1 can always win the following two-person game when playing optimally:\\
Initialization: $U_0 \gets U$, $S_0 \gets S$, $t_0 \gets t$ and $i \gets 1$\\
Round $i$ of the game:
\begin{compactenum}
	\item Player 2 selects a set of unavailable users $U'_{i-1}$ such that $|U'_{i-1}| \le t_{i-1}$
	
	$U_i \gets (U_{i-1} \setminus U'_{i-1})$ and $t_i \gets (t_{i-1} - |U'_{i-1}|)$
	\item Player 1 selects a ready step $s_{i} \in S_{i-1}$ and a user $u \in U_i$
	
	$\pi(s_{i}) \gets u$ and $S_i \gets S_{i-1} \setminus \{s_{i}\}$
	
	if $\pi$ is not a valid partial plan with respect to $s_{1}, \dots, s_{i}$, then Player 1 loses.
	\item if $S_i = \emptyset$, then Player 1 wins; otherwise, let $i \gets i+1$ and the game goes on to the next round.
\end{compactenum}
\end{df}


\begin{df}\label{def:dynres}
Given an integer $t \ge 0$, a constrained workflow authorization schema $((S, <), U, A, C)$ is {\em dynamically $t$-resilient} if and only if Player 1 can always win the following two-person game when playing optimally:\\
Initialization: $S_0 \gets S$, and $i \gets 1$\\
Round $i$ of the game:
\begin{compactenum}
	\item Player 2 selects a set $U'_{i-1}$ such that $|U'_{i-1}| \le t$
	
	$U_i \gets (U \setminus U'_{i-1})$
	\item Player 1 selects a ready step $s_{i} \in S_{i-1}$ and a user $u \in U_i$
	
	$\pi(s_{i}) \gets u$ and $S_i \gets S_{i-1} \setminus \{s_{i}\}$
	
	if $\pi$ is not a valid partial plan with respect to $s_{1}, \dots, s_{i}$, then Player 1 loses.
	\item if $S_i = \emptyset$, then Player 1 wins; otherwise, let $i \gets i+1$ and the game goes on to the next round.
\end{compactenum}
\end{df}

Wang and Li established the following results~\cite[Theorems 13--15]{WaLi10}:
\begin{itemize}
 \item determining whether a workflow is statically $t$-resilient is NP-hard (and is in coNP$^{\rm NP})$;
 \item determining whether a workflow is decrementally $t$-resilient is PSPACE-complete; and
 \item determining whether a workflow is dynamically $t$-resilient is PSPACE-complete.
\end{itemize}
The proof of the first result follows from the fact that WSP itself corresponds to the special case $t = 0$.
The remaining results are proved by establishing a reduction from the {\sc Quantified Satisfiability} problem, which is known to be PSPACE-complete.
Note that if a workflow is dynamically $t$-resilient then it is decrementally $t$-resilient; and if a workflow is decrementally $t$-resilient then it is statically $t$-resilient. 
 
\subsection{Fixed-parameter tractability results}\label{sec:res:wangli}

The moves of Player 1 in Definitions~\ref{def:decrres} and~\ref{def:dynres}, after an execution of the game, give a linear extension $[s_1, \dots, s_k]$ of $(S, <)$. On the other hand, the moves of Player 2 give a particular sequence of user sets. Using these ideas, we now give equivalent definitions of $t$-resiliency that are more suitable for developing the algorithm we present in Theorem \ref{thm:fptresiliency}. 
To do so, we first introduce the notion of a $t$-close userset family and then define what it means for a plan to be compatible with a linear extension and a $t$-close userset family.

 \begin{df}
 Let $W=((S, <), U, A, C)$ be a constrained workflow authorization schema, and $t \le |U|$. A family $F=\{U_i\}_{i \in [k]}$ is called a {\em $t$-close userset family} of $U$ if $U_i \subseteq U$ and $|U| - |U_i| \le t$ for all $i \in [k]$.
 \end{df}

 \begin{df}
 Let $W = ((S, <), U, A, C)$ be a constrained workflow authorization schema, $t \le |U|$, $F=\{U_i\}_{i \in [k]}$ be a $t$-close userset family of $U$ and $\ell = [s_1, \dots, s_k]$ be a linear extension of $(S, <)$. We say that a plan $\pi : S \rightarrow U$ is {\em $(F,\ell)$-compatible} if for every $i \in [k]$, $\pi(s_i) \in U_{i}$.
 \end{df}
 
 The following result demonstrates that dynamic $t$-resiliency can be defined in terms of compatible plans and linear extensions (rather than two-player games).
 This turns out to be a more convenient formulation for determining whether a workflow is dynamically $t$-resilient.

\begin{thm}
 Let $W=((S, <), U, A, C)$ be a constrained workflow authorization schema, and $t \le |U|$. $W$ is dynamically $t$-resilient if and only if for every $t$-close userset family $F=\{U_i\}_{i \in [k]}$ of $U$, there exists a linear extension $\ell = [s_1, \dots, s_k]$ of $(S, <)$ and a valid plan for $W$ which is $(F,\ell)$-compatible. 
 \end{thm}
 \begin{proof}
Suppose that $W$ is dynamically $t$-resilient, and suppose we are given a $t$-close userset family $F=\{U_i\}_{i \in [k]}$ of $U$. Then by assumption, if we run the game with Player 2 making choices corresponding to $F$, then there exists a sequence of moves for Player 1, that is, a linear extension $\ell = [s_1, \dots, s_k]$ of $(S, <)$ together with a valid plan $\pi : S \rightarrow U$ such that $\pi(s_{i}) \in U_i$. In other words, $\pi$ is $(F, \ell)$-compatible.

Conversely, moves of Player 2 can be represented by a $t$-close userset family $F=\{U_i\}_{i \in [k]}$. By hypothesis, there exists a linear extension $\ell = [s_1, \dots, s_k]$ of $(S, <)$ and a valid plan $\pi : S \rightarrow U$ which is $(F, \ell)$-compatible, i.e. such that Player 1 can win by picking step $s_{i}$ at iteration $i$.
 \end{proof}
 
It is easily seen that similar statements can be made for static and decremental resiliency by selecting appropriate conditions on the $t$-close userset families.
 \begin{thm}
 Let $W=((S, <), U, A, C)$ be a constrained workflow authorization schema, and $t \le |U|$. $W$ is statically $t$-resilient if and only if for every $t$-close userset family $F=\{U_i\}_{i \in [k]}$ of $U$ such that $U_i = U_j$ for all $i, j \in [k]$, there exists a linear extension $\ell = [s_1, \dots, s_k]$ of $(S, <)$ and a valid plan for $W$ which is $(F,\ell)$-compatible. 
 \end{thm}
 
  \begin{thm}
 Let $W=((S, <), U, A, C)$ be a constrained workflow authorization schema, and $t \le |U|$. $W$ is decrementally $t$-resilient if and only if for every $t$-close userset family $F=\{U_i\}_{i \in [k]}$ of $U$ such that $U_{i+1} \subseteq U_i$ for all $i \in [k-1]$, there exists a linear extension $\ell = [s_1, \dots, s_k]$ of $(S, <)$ and a valid plan for $W$ which is $(F,\ell)$-compatible.
 \end{thm}

A na\"ive approach for deciding whether a workflow is $t$-resilient would be to enumerate all possible scenarios (i.e. all possible $t$-close userset families and all linear extensions of $(S, <)$), and then test whether a valid plan is compatible with this scenario. However, such an enumeration would be highly inefficient from a computational point of view, as it would lead to an algorithm running in time $O^*(n^{tk})$.

We now show that it is sufficient to enumerate all $t$-close userset families of a {\em subset} of users whose size depends on $k$ and $t$ only. 
Using this result, we are able to prove that deciding whether a workflow is (statically, decrementally or dynamically) $t$-resilient is FPT parameterized by $k+t$. In other words, if both the number of steps and the maximum number of unavailable users are small, we obtain an efficient algorithm.
It seems reasonable to assume that $t$ will be small relative to $n$ in practice. 

\begin{thm}\label{thm:fptresiliency}
We can decide whether a constrained workflow authorization schema \mbox{$W=((S, <), U, A, C)$} is dynamically $t$-resilient for some $t \in \mathbb{N}$ in time $O^*(f(k, t))$ for some computable function $f$ if all constraints in $C$ are UI.
\end{thm}

It is important to note that all constraints considered by Wang and Li are UI.
In fact, the class of UI constraints is far larger than the constraints considered by Wang and Li in their work.

\begin{proof}[Proof of Theorem~\ref{thm:fptresiliency}]
The first step of the algorithm consists of reducing the set of users to a size depending on $k$ and $t$ only.
To do so, for any $p \le k$ and any partition $P$ of $S$ into $p$ non-empty disjoint subsets $\{T_1, \dots, T_p\}$, we construct a bipartite graph $G_P$ with partite sets $\{T_1, \dots, T_p\}$ and $U$. 
The adjacencies are defined by the authorization policy: for any $i \in [p]$ and any $u \in U$, we connect $T_i$ and $u$ if and only if $(u, s) \in A$ for all $s \in T_i$. Moreover, for all $i \in [p]$, we note $N(T_i) = \{u \in U: \forall s \in T_i, ~ (u, s) \in A\}$. 
We then proceed to a marking step, which results in a set of marked users $U_m \subseteq U$: 
for every $i \in [p]$, if $|N(T_i)| < k+t$, then mark all users in $N(T_i)$. Otherwise, mark an arbitrarily chosen set of $k+t$ users of $N(T_i)$. 
Doing this for any partition $P$, it is clear that the total number of marked users is at most $k (k+t) {\cal B}_k$, and thus $|U_m| = O(2^{k \log_2 k} k (k+t))$. 

For the remainder of the proof, let $A_m$ be the restriction of $A$ to $U_m$, i.e. $A_m = A \cap (U_m \times S)$, and $W'=((S, <), U_m, A_m, C)$. We also denote $\bar{U}_m = U \setminus U_m$.

\begin{lemma}\label{lemma:equivsol}
$W$ is dynamically $t$-resilient if and only if $W'$ is dynamically $t$-resilient.
\end{lemma}
\begin{proof}
Suppose that $W$ is dynamically $t$-resilient, and let $F=\{U_i\}_{i \in [k]}$ be  a $t$-close userset family of $U_m$. We set $U_i' = U_i \cup \bar{U}_m$ for each $i \in [k]$, and $F' = \{U_i'\}_{i \in [k]}$. One can see that $F'$ is a $t$-close userset family of $U$, and thus there exists a linear extension $\ell$ of $(S, <)$ and a valid $(F', \ell)$-compatible plan $\pi$ for $W$. Let us denote by $P(\pi) = \{T_1, \dots, T_p\}$ the pattern of $\pi$, with $p \le k$. For all $i \in [p]$, if $\pi(T_i) \in U_m$, then define $\pi'(T_i) = \pi(T_i)$. If $\pi(T_i) \notin U_m$, then it means that $\pi(T_i) \in U_{i}' \cap \bar{U}_m$, and, by construction of the marking algorithm, it implies that $|N(T_i) \cap U_m| = k+t$, and thus $|N(T_i) \cap U_{i}| \ge k$. Hence, there must exist $u \in N(T_i) \cap U_{i}$ such that $\pi(T_j) \neq u$ for all $j \neq i$, and we set $\pi'(T_i) = u$. By doing this, we can check that $\pi'$ is a valid plan for $W'$ that is $(F,\ell)$-compatible.

Conversely, suppose that $W'$ is dynamically $t$-resilient, and let $F=\{U_i\}_{i \in [k]}$ be a $t$-close userset family of $U$. In this case, we define $U'_i = U_i \cap U_m$ for every $i \in [k]$, and $F' = \{U'_i\}_{i \in [k]}$. One can verify that $|U'_i| \ge |U_m|-t$, and thus $F'$ is a $t$-close userset family of $U_m$. Thus, there must exist a linear extension $\ell$ of $(S, <)$ together with a valid plan $\pi$ for $W'$ which is $(F,\ell)$-compatible, and it is easily seen that this plan is both valid for $W$ and $(F',\ell)$-compatible.
\end{proof}

Since the size of $U_m$ now depends on $k$ and $t$ only, we can enumerate all $t$-close userset families of $U_m$: for each such family $F$, we test whether there exists a linear extension $\ell$ of $(S, <)$ (by enumerating all of them) and a valid plan for $W'$ which is $(F,\ell)$-compatible in time $O(g(k, t))$ for some computable function $g$. Since the marking algorithm runs in time $O^*({\cal B}_k)$, the total running time is $O^*(f(k, t))$ for some computable function $f$. Correctness of the algorithm is ensured by Lemma~\ref{lemma:equivsol}.
\end{proof}

It is easily seen that the previous algorithm can be adapted to the decremental and static notions of resiliency. Hence, we have the following:

\begin{thm}
We can decide whether a constrained workflow authorization schema $W=(\mbox{$(S, <)$}, U, A, C)$ is decrementally $t$-resilient (resp. statically $t$-resilient) in time $O^*(f(k, t))$ for some computable function $f$ if all constraints in $C$ are UI.
\end{thm}

\subsection{A new model for resiliency}

Mace \etal argue that knowing whether a workflow is dynamically $t$-resilient or not is unlikely to be of much use in the real world~\cite{MaMoMo14}.
They suggest it would be more valuable to quantify the extent to which a workflow was resilient.
Accordingly, they introduce metrics for ``quantitative resiliency'' and used Markov decision processes to compute these metrics.
Informally, they consider every pair of $t$-close userset family and plan to compute these metrics, making their approach computationally expensive.
Their approach also appears to assume that the \emph{set of unavailable users} is selected uniformly at random from all possible subsets of available users. 
This means, in particular, that the probability of no users being unavailable is the same as $t$ specific users being unavailable -- an assumption that seems unlikely to hold in practice.

While agreeing with the motivation for studying quantitative resiliency, we adopt a different approach and one which enables us to treat questions about quantitative resiliency as instances of \BOWSP.
Our basic approach is to identify three aspects that affect resiliency: the authorization policy, ``known'' unavailability of users, and ``unknown'' unavailability of users.
We start with a constrained workflow authorization schema  $W=((S,<),U,A,C)$, where $A$ is given as a function from $S \times U$ to $\{0, 1\}$ (i.e. $A(s,u) = 1$ if and only if user $u$ is authorized to perform step $s$). 
We are also given a description of user availability, comprising a deterministic component (the ``known'' unavailability) and a probabilistic component (the ``unknown unavailability'').

The deterministic component is specified as a function $\alpha : S \times U \rightarrow \set{0,1}$, where $\alpha(s,u) = 1$ if and only if $u$ is available to execute step $s$.
Such a component may be derived from inspection of user calendars and take into account activities such as holidays and time away from the office on business.
(Of course, defining $\alpha$ will also require some assumptions to be made about when the steps in a workflow instance will be executed, which may not always be possible.
In such cases, we can simply set $\alpha(s,u) = 1$ for all $s$ and $u$.)
 
The probabilistic component is specified as a function $\rho : S \times U \rightarrow [0,1]$, where $\rho(s,u)$ denotes the probability that $u$ is not available to execute step $s$. 
Values of $\rho(s,u)$ could be derived from historical information about user availability, such as (unplanned) absences from work.
Given $W$, $\alpha$, $A$ and $\rho$, we define
 \[
  \omega(s,u) = %
   \begin{cases}
     M & \text{if $A(s,u) \cdot \alpha(s,u)=0$}, \\
     \rho(s,u) & \text{otherwise},
   \end{cases}
 \]
where $M$ is a large positive constant.
Finally, for a plan $\pi$, we define \[ \omega_A(\pi) = \sum_{s \in S} \omega(s, \pi(s)). \]
In other words, assigning an unauthorized or unavailable user results in a weight of at least $M$, while assigning an authorized user incurs a weight corresponding to the likelihood of the user being unavailable. 

Clearly $\omega_A$ respects the definition of an authorization weight function as described in Section~\ref{sec:wwsp}.
In addition to this authorization function, we can define a constraint weight function $\omega_C : \Pi \rightarrow \mathbb{Q}^+$, as we did in Section~\ref{sec:wwsp}.
We may now solve the resulting instance of \BOWSP with one of the three approaches mentioned earlier.
 
We now consider how the result of solving such an instance of \BOWSP may be interpreted and provide useful information to workflow administrators.
For a plan $\pi$, let $X(\pi)$ be a random variable equal to the number of unexecuted steps due to users being unavailable. 
Notice that by definition of expectation, if $\omega_A(\pi)$ is less than $M$, then $\mathbb{E}[X(\pi)]=\omega_A(\pi)$ and thus by Markov's inequality\footnote{Markov's inequality states that if $X$ is a non-negative random variable and $a>0$, then $\mathbb{P}(X\ge a)\le \mathbb{E}[X]/a$~\cite{Weisstein}.} and the fact that $X(\pi)$ is integer-valued, the probability that  all steps will be executed by $\pi$ is 
 \begin{align*} \mathbb{P}(X(\pi) = 0) &= 1-\mathbb{P}(X(\pi)\ge 1) \\ &\ge  1-\omega_A(\pi). \end{align*}

Now consider our purchase order workflow example.
For simplicity, we will assume that $\rho(s,u)$ does not depend on $s$; that is, $\rho(s,u)=\rho(s',u)$ for all $s,s' \in S$ and we may simply write $\rho(u)$.
Illustrative values of $\rho(u)$ are shown in Figure~\ref{subfig:user-probabilities}, and the function $\omega$ derived from those probabilities (using Figures~\ref{subfig:authorization-relation} and~\ref{subfig:user-probabilities}) is shown in Figure~\ref{subfig:omega-for-resiliency-example}. 

\begin{figure}[h]
\begin{subfigure}{\textwidth}\centering
 \begin{tabular}{|*{8}{>{$}r<{$}}|} 
 \hline
  u_1 & u_2 & u_3 & u_4 & u_5 & u_6 & u_7 & u_8 \\
 \hline
  ~0.01 & ~0.06 & ~0.03 & ~0.05 & ~0.07 & ~0.05 & ~0.06 & ~0.01 \\
 \hline
 \end{tabular}
\caption{$\rho : U \rightarrow [0,1]$}\label{subfig:user-probabilities}
\end{subfigure}

\vspace*{.5\baselineskip}

\begin{subfigure}{\textwidth}\centering
  \begin{tabular}{|>{$}r<{$}|*{6}{>{$}r<{$}}|}
  \hline
   & s_1 & s_2 & s_3 & s_4 & s_5 & s_6 \\
  \hline
   u_1 & 0.01 & M & 0.01 & 0.01 & M & M \\
   u_2 & 0.06 & M & 0.06 & 0.06 & M & M \\
   u_3 & 0.03 & M & 0.03 & M & M & M \\
   u_4 & 0.05 & M & 0.05 & M & M & M \\
   u_5 & 0.07 & M & 0.07 & M & M & M \\
   u_6 & M & 0.05 & 0.05 & M & 0.05 & M \\
   u_7 & M & M & 0.06 & 0.06 & 0.06 & M \\
   u_8 & M & M & M & M & 0.01 & 0.01 \\
  \hline
  \end{tabular}
\caption{$\omega : S \times U \rightarrow \mathbb{Q}^+$}\label{subfig:omega-for-resiliency-example}
\end{subfigure}
\caption{Assigning probabilistic weights}\label{fig:assigning-probabilistic-weights}
\end{figure}

For ease of illustration we assume that the penalty of breaking each constraint is the same and that we allow at most one constraint to be violated. 
We also assume that $\alpha(s,u) = 1$ for all $s$ and $u$.
Then it makes sense to seek a Pareto optimal plan (since the above assumptions essentially fix the maximum penalty due to the violation of constraints).
The plan $\tau$ that minimizes $\omega_A$ is given by $\tau(s_1)=\tau(s_3)=\tau(s_4)=u_1$, $\tau(s_2)=u_6$, $\tau(s_5)=\tau(s_6)=u_8$ and the expected value of the number of unexecuted steps $\omega_A(\tau)$ equals $0.10$. 
Thus, the probability that all steps will be executed is at least $0.90$.\footnote{Markov's inequality is known to provide only a rough bound and so the probability that all steps will be executed is likely to be even larger.} 
Now assume that we do not allow any constraint to be violated.
Then $\tau'(s_1)=\tau'(s_3)=u_3$, $\tau'(s_4)=u_1$, $\tau'(s_2)=u_6$, $\tau'(s_5)=\tau'(s_6)=u_8$ is a minimum authorization weight plan and the expected value of the number of unexecuted steps equals $0.14$, implying that the probability that all steps will be executed is at least $0.86$.
Notice that the probability that all steps will be executed is lower when we don't allow any constraints to be violated (as one would expect).
Hence, we believe that the model presented in this section may fit well
to many practical situations, and that the results provided by a solution to
BO-WSP may provide useful information to workflow administrators.

\section{Concluding remarks}\label{sec:conclusion}

In this paper, we have generalized the {\sc Valued WSP} problem~\cite{CrGuKa15}, by defining the {\sc Bi-objective WSP} (\BOWSP).
In doing so, we have introduced a tool for analyzing a variety of problems relating to workflow satisfiability.
Moreover, we have developed fixed-parameter tractable algorithms to solve \BOWSP.
This means that we are able to solve many instances of \BOWSP in which the parameters take values that one would expect to see in practice (namely, the numbers of users is large compared to the number of steps).

Our work improves on existing work in a number of ways and provides a single framework within which a variety of problems may be analyzed and solved.
In particular, we are able to solve {\sc Valued WSP} (and generalizations thereof), the Cardinality-Constrained Minimum User Problem~\cite{Roy2015} (and generalizations thereof), and questions about workflow resiliency.
We have also established that the decision problems associated with resiliency due to Wang and Li are fixed-parameter tractable.

We have taken inspiration from the work of Mace \emph{et al.}~\cite{MaMoMo14} by developing a new model for quantifying workflow resiliency in terms of the expected number of steps that will be executed for a given plan, based on probabilities associated with the (un)availability of users.
We believe our model for user unavailability is more realistic than that of Mace \emph{et al}.
It is certainly the case that our approach leads to more efficient methods of computing the expected number of steps that will be executed.
Nevertheless, Mace \emph{et al.} introduced a number of useful metrics that cannot be computed using \BOWSP.
Unfortunately,  these metrics are very expensive to compute, because it is necessary to consider every possible plan and every possible sequence of usersets that may be available.

There are several opportunities for further research in this area.
For most workflow specifications it is not necessary to consider the control flow (specified in our model by the partial ordering on the set of steps)~\cite{CrGu13}.
However, some workflow specifications may contain sub-workflows that are iterated and problems analogous to WSP in this context are known to be hard~\cite{BaBuKa12}.
We intend to investigate whether our methods can be extended to cater for such workflows.

We also hope to extend our research into quantitative workflow resiliency.
In particular, we would like to develop metrics that provide useful information in practice (like those of Mace \emph{et al.}) but can be computed using \BOWSP (for appropriate choices of weight functions).

\paragraph{Acknowledgements.}
This research was partially supported by an EPSRC grant EP/K005162/1. 
Gutin's research was partially supported by Royal Society Wolfson Research Merit Award.
Karapetyan's research was partially supported by EPSRC grant EP/H000968/1.

\bibliography{refs}
\bibliographystyle{abbrv}

\newpage
\appendix

\section{Proof of Theorem~\ref{thm:tightness}}\label{app:pareto-front-may-have-B-k-points}

\begin{proof}
Let $S$ be the set of $k$ steps, and consider all unordered pairs of steps in an arbitrary order $c_1, \dots, c_{\binom{k}{2}}$. We define a separation-of-duty constraint for every such pair, and define $\omega_{c_i}(\pi) = 2^i$ for all $i \in \{1, \dots, {k \choose 2}\}$ and all plans $\pi$ not satisfying $c_i$. Then, for all $T \subseteq S$, introduce a user $u_T$, and denote by $U$ the set of all such users. To define the authorization weights, we introduce, for all $T \subseteq S$, $cut(T) = \{\{s, s'\}: s \in T, s' \in S \setminus T\}$.
We then define, for all $T, Q \subseteq S$:
\[
 \omega(Q,u_T) = 
  \begin{cases}
   \frac{1}{2}\sum_{c_i \in cut(T)} 2^i  & \text{if $Q=T$,} \\
  	\infty & \text{otherwise}.
  \end{cases}
\]
 And, finally, recall that $\omega_A(\pi) = \sum_{u \in U} \omega(\pi^{-1}(u), u)$ for any plan $\pi : S \rightarrow U$.
It is easy to see that, by construction, all patterns have different constraint weights; and for a fixed pattern $P = \{T_1, \dots, T_q\}$, $q \le k$, there exists a unique complete plan $\pi_P$ with a finite authorization weight (the one that assigns $u_T$ to steps in $T$). Moreover, remark that $\omega_C(\pi_P)$ equals the sum of the weights of all constraints whose both steps are withing the same set, while $\omega_A(\pi_P)$ equals the sum of the weights whose steps lie in different sets. In other words, we have $\omega_A(\pi_P) = M - \omega_C(\pi_P)$, where $M = \sum_{i=1}^{k \choose 2} 2^i$. Hence, the Pareto front consists of one point per pattern, and is thus of size ${\cal B}_k$.
\end{proof}

\section{The branching heuristic}\label{sec:bh}

 Recall that the branching heuristic of PBB is responsible for selecting step $s$ in each node of the search tree; this step is then used to generate children of the current pattern.
 As is typical in tree search methods, the aim is to select a step such that the most constrained parts of the solution are defined as early as possible.
 This tends to dramatically reduce the size of the search tree and hence the running time of the algorithm.
 
 Let $P$ be the current pattern, and $S' \subset S$ be the set of currently assigned steps.
 Let $C(s)$ be the set of constraints that include step $s$ in their scope $T_c$.
 Let $C_{\neq}(s)$, $C_{\le}(s)$ and $C_{\ge}(s)$ be the sets of separation-of-duty, at-most and at-least constraints that include step $s$ in their scopes.
 Then the branching heuristic of PBB selects a step $s \not\in S'$ that maximizes the following function:
\begin{equation}
	\label{eq:rho}
	\rho(P, s) =  \phi(s) + \sum_{c \in C(s)} \rho^c \big( |\set{B \in P :\ B \cap T_c \neq \emptyset} |,\; |S' \cap T_c| \big) \,,
\end{equation} 
where $T_c$ is the scope of constraint $c \in C$, $\rho^c(m, t)$ is a table function and $\phi(s)$ is a precalculated value reflecting the relative importance of a step (in the absence of knowledge about the current pattern):
\begin{multline}
	\label{eq:psi}
	\phi(s) = \sum_{c \in C(s)} \psi_c
	+ 50 \cdot | \set{ c \in C_{\neq}(s), c' \in C_{\le}(s) : T_c \subset T_{c'} } | \\
	+ 50 \cdot | \set{ c \in C_{\le}(s), c' \in C_{\ge}(s) :\ |T_c \cap T_{c'}| \ge 3 } |
	+ | \set{ u \in U :\ \omega(\set{s}, u)) > 0 } |.
\end{multline}
 The value $\psi_c$ depends on the type of constraint $c$; for a separation-of-duty constraint it is 1, and for at-most or at-least constraints it is 5.
 The table function $\rho^c(m, t)$ reflects the degree of constraintness of constraint $c$ if $t$ of its steps are already assigned, and the number of already assigned users is $m$;
 Figure~\ref{fig:constraint-penalties-appendix} provides illustrative values of $\rho^c(m, t)$ for at-most-3 out of 5 and at-least-3 out of 5 constraints respectively, where $m$ is the number of distinct users assigned to the scope of constraint $c$, and $t$ is the number of already assigned steps in the scope of $c$.
 	$\rho^c(m, t) = 0$ for all the separation-of-duty constraints.

\begin{figure}[htb]
 \centering
{
	\begin{tabular}{@{} crrrr @{}}
		\toprule
 		$m~\backslash~t$ & 1 & 2 & 3 & 4 \\
		\midrule
		1	&	0	&	0	&	0	&	0 \\
		2	&	--	&	0	&	20	&	0 \\
		3	&	--	&	--	&	500	&	500 \\
		4	&	--	&	--	&	--	&	200 \\
		\bottomrule
	\end{tabular}
	 }
 \qquad
{
	\begin{tabular}{@{} crrrr @{}}
		\toprule
 		$m~\backslash~t$ & 1 & 2 & 3 & 4 \\
		\midrule
		1	&	0	&	0	&	10	&	5 \\
		2	&	--	&	0	&	0	&	20 \\
		3	&	--	&	--	&	0	&	0 \\
		4	&	--	&	--	&	--	&	0 \\
		\bottomrule
	\end{tabular}
 }

 \caption{Examples of constraint penalties for at-most-3 and at-least-3 constraints}\label{fig:constraint-penalties-appendix}
\end{figure}

 While it is difficult to justify every term in $\rho(\mathcal{P}, s)$ and therein referenced functions, the idea is to account for each easy-to-compute interaction between step $s$ and constraints and authorizations, while keeping as many of the terms as possible in $\phi(s)$ because $\phi(s)$ can be precomputed before the tree search.
 Each of the terms is provided with a coefficient which is then obtained by automated parameter tuning, an approach in which we formulate an optimization problem of finding the optimal parameter values and then solve it heuristically.
 
\section{MIP encodings}\label{sec:mip-encodings}

Here we describe how the constraints and users present in our test instances can be encoded for the MIP-based algorithm.
 The encoding for a staff member $u$ is straightforward:
 $$
 \omega_u(x) = \sigma_u \sum_{s \in B_u} x_{su} + M \sum_{s \in T\setminus (A_u\cup B_u)} x_{su}\,.
 $$
 The encoding for a consultant $u$ requires additional binary decision
variables $p_u$ and $q_u$, where $p_u$ is 1 if and only if at least one step
is assigned to $u$, and $q_u$ is 1 if and only if at least one step from
outside $A_u$ is assigned to $u$:
$$
p_u \ge x_{su} \qquad \forall s \in S
\qquad
\text{and}
\qquad
q_u\ge x_{su} \qquad \forall s \notin A_u \,.
$$
 Then $\omega_u(x)$ for a consultant is defined as
$$
\omega_u(x) = \sigma_u p_u + (M - \sigma_u) q_u \,.
$$

 For a separation-of-duty constraint $c$ with scope $\set{s, t}$ we also need an additional binary decision variable $p$:
 $$
 p_{c} \ge x_{su} + x_{tu} - 1 \qquad \forall u \in U \,,
 $$
 which is then used in the objective function as follows:
 $$
 \omega_c(x) = Mp_c \,.
 $$
 
 Let $c$ be an at-least-3 constraint with scope $T$.
 We introduce binary variables $z_u$, $u \in U$, such that $z_{uc} = 1$ if and only if user $u$ is assigned to at least one step in $T$:
 $$
 z_{uc} \le \sum_{s \in T} x_{su} \qquad \forall u \in U \,.
 $$
 We also need binary variables $p^i_c$ for $i = 1, 2$ such that $p^i_c = 1$ if and only if there are no more than $i$ users assigned to all the steps in $T$:
 $$
 \sum_{u \in U} z_{uc} + \sum_{i = 1}^2 p^i_c \ge 3 
 \qquad 
 \text{and} 
 \qquad
 p^1_c \le p^2_c \,.
 $$
 Then 
 $$
 \omega_c(x) = (\sigma^1 - \sigma^2) \cdot p^1_c + \sigma^2 \cdot p^2_c \,,
 $$
 where $\sigma^i$ is the penalty for assigning exactly $i$ users to the scope $T$.

 Encoding of an at-most-3 constraint $c$ with scope $T$ also requires variables $z_u$, but since we are now maximizing them, they need to be constrained from above:
 $$
 z_{uc} \le x_{su} \qquad \forall s \in T,\ \forall u \in U \,.
 $$
 We again use binary variables $p^i_c$ but for $i = 4, 5$:
 $$
 \sum_{u \in U} z_{uc} - \sum_{i = 4}^5 p^i_c \le 3 
 \qquad 
 \text{and} 
 \qquad
 p^5_c \le p^4_c \,.
 $$
 Then 
 $$
 \omega_c(x) = (\sigma^5 - \sigma^4) \cdot p^5_c + \sigma^4 \cdot p^4_c \,,
 $$
 where $\sigma^i$ is the penalty for assigning exactly $i$ users to the scope $T$.

Notice that some of the auxiliary variables introduced above can actually be made continuous without loss of correctness. Normally this will allow MIP solver to perform more efficiently.

\end{document}